\documentclass[
  aip,
  jcp,
  letterpaper,
  amsmath,
  amssymb,
  reprint,
  longbibliography,
  floatfix
]{revtex4-2}

\usepackage{mathtools,bm,amsthm}
\usepackage{booktabs}
\usepackage{array}
\usepackage{graphicx}
\usepackage[nopatch=footnote]{microtype}
\usepackage[capitalise]{cleveref}
\crefname{equation}{Eq.}{Eqs.}
\crefname{figure}{Fig.}{Figs.}
\crefname{section}{Sec.}{Secs.}
\Crefname{equation}{Equation}{Equations}
\Crefname{figure}{Figure}{Figures}
\Crefname{section}{Section}{Sections}

\begingroup\fontfamily{cmss}\selectfont\endgroup
\DeclareFontShape{OT1}{cmss}{bx}{it}{<->ssub*cmss/bx/n}{}

\newtheorem{theorem}{Theorem}
\newtheorem{proposition}[theorem]{Proposition}
\newtheorem{lemma}[theorem]{Lemma}
\newtheorem{corollary}[theorem]{Corollary}
\newtheorem{remark}[theorem]{Remark}

\newcommand{\cB}{\mathcal B}
\newcommand{\cC}{\mathcal C}
\newcommand{\cF}{\mathcal F}

\newcommand{\cH}{\mathcal H}
\newcommand{\cL}{\mathcal L}
\newcommand{\cM}{\mathcal M}
\newcommand{\cP}{\mathcal P}
\newcommand{\cQ}{\mathcal Q}

\newcommand{\cV}{\mathcal V}
\newcommand{\cW}{\mathcal W}
\newcommand{\Tr}{\operatorname{Tr}}

\newcommand{\spec}{\operatorname{spec}}

\newcommand{\fall}[2]{\left(#1\right)_{\underline{#2}}}
\newcommand{\id}{\mathbb I}
\newcommand{\R}{\mathbb R}
\newcommand{\C}{\mathbb C}

\newcommand{\pN}{{\pi_2^{[N]}}}
\newcommand{\pNadj}{(\pi_2^{[N]})^*}
\newcommand{\SigmaN}{{\Sigma_q^{[N]}}}
\newcommand{\coneN}{{\cC_{2,q}^{[N]}}}

\begin{document}

\title{Spurious zero-dispersion continua in full $(2,q)$-positivity: A linear hierarchy-depth lower bound from an exactly solvable counting operator}

\author{Maho Nakata}
\email{maho@riken.jp}
\affiliation{Computational Condensed Matter Physics Laboratory, RIKEN Pioneering Research Institute, 2-1 Hirosawa, Wako, Saitama 351-0198, Japan}
\date{August 30, 2026}

\begin{abstract}
Variational two-particle reduced-density-matrix methods impose hierarchies of semidefinite constraints. We ask whether a fixed level of such a hierarchy can recover even an elementary many-fermion spectrum, using the dispersion test, in which a value is accepted as an eigenvalue when the minimized dispersion vanishes. We study an exactly solvable operator that counts fermions in one half of the spin orbitals.  It is a sum of mutually commuting occupation numbers, so its spectral problem is classical, a Boolean occupation-number problem with integer spectrum, and the obstruction found here is not a genuinely quantum one.  We determine the complete set of accepted values at every level of the full hierarchy of particle--hole semidefinite constraints. Besides the integers, this set contains an entire interval of noninteger values unless the hierarchy level grows linearly with the particle number, and we give the exact threshold. The proof constructs explicit $q$-positive reduced density matrices that are diagonal in the occupation basis, obtained by continuing the occupation probabilities of a fixed-particle-number state to a noninteger particle number.  The fermionic content of the work is that classical semidefinite constraints on occupation correlations, supplemented by nonnegative occupation-pattern probabilities, imply positivity of every particle--hole block, so that the noncommuting fermionic conditions remove no false value. For this counting model the two relaxations have the same dispersion zero set, not necessarily the same feasible states, and a converse argument shows that no other false eigenvalues occur. The limitation is specific to the hierarchy: two-body inequalities that encode the integer spectrum would remove every false value but are not generated at a fixed level. By Erdahl's duality theorem for the lower-bound method, whose hypotheses we check in the fixed-particle-number setting, the result applies equally to the polar formulation, the full $(2,q)$-positivity conditions of the constructive solution of the $N$-representability problem.
\end{abstract}


\maketitle

\section{Introduction}

The fermionic $N$-representability problem asks which low-order reduced
density matrices (RDMs) arise from an $N$-particle state
\cite{Coleman1963,Erdahl1978}.  A sensitive test of an approximate
representability condition is the \emph{approximate spectrum} of a one-body
operator $h$, introduced by Erdahl and Grudzi\'nski
\cite{ErdahlGrudzinski1978}: minimize the dispersion $(h-\lambda\id)^2$ over
the approximate set of 2-RDMs and accept $\lambda$ as an eigenvalue when the
minimum vanishes.  With exact representability the accepted values are the
eigenvalues of $h$; with an outer approximation the set can be larger.
Erdahl already recorded that for the number operator on a Fock space with
$r$ orbitals it becomes $\{0,r\}\cup[1,r-1]$
\cite[Sec.~6, Theorem~6.1]{Erdahl1978}, a continuum of false eigenvalues.

We determine this set exactly at every level of the complete hierarchy of
positivity conditions for a counting operator $K$ at fixed particle number
$N$: the number of fermions in one half of $2N$ spin orbitals.  Its spectrum
is $\{0,1,\ldots,N\}$.  Because $K$ is a sum of mutually commuting
occupation numbers, deciding which values have zero dispersion is a
classical problem about Boolean occupation numbers $n_i\in\{0,1\}$, and the
hierarchy-depth bound below is not a genuinely quantum obstruction.  The
fermionic question is whether the particle--hole conditions of the
hierarchy, which involve noncommuting operators and have no classical
counterpart, remove any of the classical false values; we show that they
do not.  Full $q$-positivity, that is, positive
semidefiniteness of all particle--hole metric blocks up to order $q$, such
as $D_3,E_3,F_3,Q_3$ at $q=3$, and the equivalent full $(2,q)$-positivity
accept exactly the values
\[
 \{0,\ldots,q-2\}\cup[q-1,N-q+1]\cup\{N-q+2,\ldots,N\}
\]
(\cref{thm:main}).  The central interval contains noninteger values
whenever $N\ge2q-1$.  Recovering the spectrum therefore requires
$q\ge\lceil(N+2)/2\rceil$, a level linear in $N$, although the exact
spectrum is known without any computation.  For comparison, the level
$q=N$ is exact for every two-body Hamiltonian, since $D_N$ is then the
$N$-electron density matrix itself; by the primal--polar equivalence
\cite{Erdahl2007} (\cref{prop:projection}) the same holds for full
$(2,N)$-positivity, and Ref.~\onlinecite{RubinDePrince2026}
constructs the corresponding witnesses explicitly.

The proof constructs explicit interpolated reduced density matrices.  They are
diagonal in the occupation basis, and their entries are the occupation
probabilities of a state with a fixed number of electrons in one half of
the orbitals, continued to a noninteger particle number $\lambda$.  For
$N=5$ and $q=3$, a diagonal 2-RDM specified by three occupation and
pair-occupation values passes the $D$, $Q$, $G$, $T_1$, $T_2$, and all
other $(2,3)$ conditions while giving zero dispersion at $\lambda=5/2$.
On the central interval they are $q$-positive RDMs in the standard sense:
they satisfy every $q$-positivity condition but are not $N$-representable.
We call them interpolated RDMs because the same formula, outside that
interval, defines RDMs that are not $q$-positive.  They are the fermionic
counterpart of the pseudo-moments of a known classical
construction used for sum-of-squares lower bounds on problems with $0/1$
variables \cite{Lasserre2001,Grigoriev2001,Lee2016,Potechin2019}.  We show
that classical positivity of occupation-number polynomials of degree
$q-1$, together with nonnegative occupation-pattern probabilities on at
most $q$ orbitals, suffices for full fermionic $q$-positivity.  For this
model these conditions give the same dispersion zero set, without asserting
equality of feasible sets (Sec.~\ref{sec:relation}).  The specifically
fermionic content of this work is this transfer to all particle--hole
blocks; also new are the converse that
no other values are accepted, and the resulting exact zero set and
recovery threshold.

The limitation belongs to the hierarchy, not to the 2-RDM as such.  The
integer spectrum implies the two-body inequalities
$(K-j)(K-j-1)\ge0$; adding them would remove every false value, but for
$q-1\le j\le N-q$ they are not generated by the sums of squares of level
$q$ (Sec.~\ref{sec:integrality}).  Likewise, the result concerns the zero
set only: even when the zero set is correct, the minimized dispersion can
be too low (Sec.~\ref{sec:energy}).

This work continues a line of variational 2-RDM calculations.  Nakata
\emph{et al.} solved the $P$, $Q$, $G$ problem as a standard semidefinite
program \cite{Nakata2001,Nakata2002}; Mazziotti subsequently applied the same approach
to atoms and molecules, adding parts of three-positivity obtained by
lifting \cite{Mazziotti2002}; Zhao \emph{et al.} added the
three-index $T_1$ and $T_2$ conditions in a dual formulation
\cite{Zhao2004,Fukuda2007}; and later calculations used the generalized
$T_2'$ condition \cite{Braams2007,Nakata2008}.  In 2006 Nakata \emph{et al.}
returned to the dispersion test and found simple Hamiltonians for which
PQG$T_1T_2$ gives drastic non-representable solutions, although the exact
state can be a single determinant \cite{Nakata2006}.  Failures under
complete three-positivity are also known: Li, Liebenthal, and DePrince
showed that substantial energy errors in two- and three-dimensional
hydrogen clusters persist with full three-particle conditions (3POS)
\cite{Li2021}.  Thus failure beyond selected conditions is not new.
Here we study an analytically solvable counting model within the complete
hierarchy of the constructive solution of the $N$-representability problem
\cite{Mazziotti2012PRL,Mazziotti2012PRA,Mazziotti2023}.  The definition of
full $(2,q)$-positivity, its polar-cone formulation, and the
finite-dimensional cancellation equation \cref{eq:maz-eq} are taken from
these works, and the equivalence of the primal and polar formulations is
Erdahl's duality theorem \cite{Erdahl2007}; Sec.~\ref{sec:2q} only
restates them in the fixed-$N$ notation used here.
Our contribution is positivity of every particle--hole block at arbitrary
$q$, a necessity proof determining the entire dispersion zero set, and the
sharp hierarchy level at which the integer spectrum is recovered.

The $(2,q)$ conditions have a longer history.  In Erdahl's formulation
\cite{Erdahl1978}, conditions on the 2-RDM are two-body operators that are
nonnegative on every $N$-particle state, and natural examples are squares
$y^\dagger y$ that lie in the two-body space.  Squares of degree-two operators give $P$, $Q$, and $G$.
Erdahl also exhibited cancellation of the highest-degree terms in
symmetrized cubic squares \cite[Sec.~8]{Erdahl1978}, particular $(2,3)$
conditions in present terminology.  Zhao
\emph{et al.} formulated $T_1$ and $T_2$ as explicit matrix inequalities and
noted that they are implicit in Erdahl's analysis \cite{Zhao2004}; the
generalized $T_2'$ is due to Braams, Percus, and Zhao \cite{Braams2007}.
Mazziotti systematized the idea, allowing conic combinations of squares at
any degree $q$ whose higher-body terms cancel, into a constructive solution
of the $N$-representability problem: a complete hierarchy of $(2,q)$
conditions, of which the classes written out in
Refs.~\onlinecite{Mazziotti2012PRL,Mazziotti2012PRA} are explicit examples.
The complete hierarchy was later expressed through a single matrix
equation for the cancellation of higher-body operators
\cite{Mazziotti2023}, and extended to systems without particle-number
conservation, with the fixed-$N$ sector recovered through the
particle-number variance \cite{Mazziotti2026PRL}.  Erdahl proved that the
$k$-positivity lower-bound program and its dual spectral optimization have
optimal solutions with no duality gap \cite[Sec.~IV]{Erdahl2007}.  Rubin, Low, and DePrince
derived the dual of the variational 2-RDM program as a weighted
sum-of-squares representation of the Hamiltonian
\cite{RubinLowDePrince2026}, and Rubin and DePrince generated the complete
family of $(2,q)$ witnesses in this form \cite{RubinDePrince2026}.  Full
three-positivity with
all four blocks $D_3,E_3,F_3,Q_3$ was implemented in
Ref.~\onlinecite{Mazziotti2006}.  The convex geometry is also older: Erdahl
describes representable RDMs as projections of the positive-semidefinite
cone and their polar as positive two-body operators \cite{Erdahl1978};
Mazziotti and Erdahl systematized higher particle--hole conditions
\cite{MazziottiErdahl2001}; and Erdahl and Jin formulated truncated
positivity as positive semidefiniteness of a moment matrix with fermionic
consistency relations, together with a dual sum-of-squares representation
\cite{ErdahlJin2000Chapter,ErdahlJin2000}.  These are antecedents of the
structures used here, not new principles.

Section~II defines the hierarchy and recalls that its primal and polar
forms give the same 2-RDMs.  Section~III introduces the counting model.
Section~IV constructs the interpolated RDMs and checks them explicitly at $q=3$.
Section~V proves the exact zero set for every $q$.  Section~VI discusses
the scope of the result and its relation to classical occupation-number
relaxations.

\section{Fermionic positivity hierarchies}

\subsection{Notation and hierarchy convention}

The notation is collected in Appendix~\ref{app:notation}.  In particular,
$r$ is the number of spin orbitals, $N$ the fixed number of electrons, and
$q$ the positivity level, counting creation and annihilation factors in a
test row.  A superscript $[N]$ indicates the fixed-$N$ space, not the body
order of an operator.  We display this dependence on cones and feasible
sets, but suppress it on operator spaces and Gram maps: throughout,
$\cV_q$, $\cW_{\le q}$, and $\cM_{q,\alpha}$ refer to the same fixed $N$.

Thus $q=5$ means full five-body positivity.  A typical $D_5$ entry contains
five creation and five annihilation operators, but $q=5$ does \emph{not} mean
that the system has ten spin orbitals.  Our smallest balanced $q=5$ example
has $N=9$ fermions in $r=18$ spin orbitals, for which the $D_5$ block alone
has dimension $\binom{18}{5}=8568$.

``Full'' means that all particle--hole blocks through level $q$ and, in the
polar form, unrestricted positive-semidefinite multipliers are retained.  It
does not denote a selected subset such as $T_1$, $T_2$, tensor-factorized
subfamilies, or approximate $(2,4)$ conditions.

\subsection{CAR algebra and truncated moments}

Let the one-particle space be $\C^r$, with fermionic creation and annihilation
operators satisfying
\[
 \{a_i,a_j^\dagger\}=\delta_{ij},\qquad
 \{a_i,a_j\}=0,\qquad
 \{a_i^\dagger,a_j^\dagger\}=0.
\]
We work in the fixed-$N$ sector $\cH_N=\bigwedge^N\C^r$ with projector $P_N$
on Fock space $\mathcal F$.  A word is a product of creation and
annihilation operators, and its \emph{length} is the number of factors.
Let $\cW_{\le q}$ denote the span of the operators $CP_N:\cH_N\to\mathcal F$
with $C$ a word of length at most $q$.  For a word that changes the particle
number we mean $CP_N$, not $P_NCP_N$, which would vanish; in either case
$C^\dagger C$ is an operator on $\cH_N$.

A normalized Hermitian truncated functional $\cL$ is called
\emph{full positive through order $q$} if
\begin{equation}
 \cL(C^\dagger C)\ge0
 \qquad\text{for every }C\in\cW_{\le q},
 \label{eq:qpos}
\end{equation}
and if it obeys the fixed-particle-number consistency relations through the
available degree.  After splitting by net particle-number change and
quotienting operator directions that vanish on $\cH_N$, this is positivity
of all particle--hole metric blocks through order $q$; at order three the
new blocks are $D_3,E_3,F_3,Q_3$.

For $0<N<r$ and $1\le q\le N$, the lower-order conditions follow from the
complete set of $q$th-order conditions.  If $D$ is a word polynomial that
changes the particle number by $\Delta$, then on $\cH_N$
\begin{equation}
 \begin{aligned}
 \sum_{i=1}^r (a_iD)^\dagger(a_iD)&=(N+\Delta)\,D^\dagger D,\\
 \sum_{i=1}^r (a_i^\dagger D)^\dagger(a_i^\dagger D)&=(r-N-\Delta)\,D^\dagger D,
 \end{aligned}
 \label{eq:fixed-N-lifting}
\end{equation}
and at least one coefficient is positive, so every lower-order square is a
conic combination of squares one operator longer.  Thus our cumulative
convention is equivalent to conventional full $q$-positivity, represented
by the $q+1$ exact-order particle--hole blocks.  We denote the resulting
normalized feasible set by $\cF_q^{\mathrm{full},[N]}$, and its restriction to
the zero-, one-, and two-body operator space $\cV_2$ by
\[
 \pN:\cF_q^{\mathrm{full},[N]}\longrightarrow\cV_2^* .
\]
For a functional $\cL$ in this set we write $\langle O\rangle=\cL(O)$.
Its two-body part $D_2=\pN(\cL)$ consists of the normalization, the
1-RDM, and the 2-RDM.  It is useful to distinguish two opposite operations:
\emph{contraction} retains reduced data from a state, whereas
\emph{lifting} makes a low-body operator act on the many-electron space.
For a physical density operator $\rho_N$ on $\cH_N$, the two directions are
\[
 \rho_N\ \xrightarrow{\ \text{contraction}\ }\ (1,D_1,D_2),
\]
\[
 (h_0,h,v)\ \xrightarrow{\ \text{lifting}\ }\ H^{[N]}.
\]
For example, with antisymmetrized two-electron coefficients $v_{ij,kl}$,
\begin{equation}
 H^{[N]}=\left.
 \left(h_0\id+\sum_{ij}h_{ij}a_i^\dagger a_j
 +\frac14\sum_{ijkl}v_{ij,kl}a_i^\dagger a_j^\dagger a_l a_k\right)
 \right|_{\cH_N}.
 \label{eq:operator-lifting}
\end{equation}
The superscript does not make $H$ an $N$-body interaction: it still has
body order at most two.  With $D_{1;i,j}=\Tr(\rho_Na_i^\dagger a_j)$ and
$D_{2;ij,kl}=\Tr(\rho_Na_i^\dagger a_j^\dagger a_l a_k)$, the same energy is
\begin{align}
 \Tr_{\cH_N}(\rho_N H^{[N]})
 &=h_0+\sum_{ij}h_{ij}D_{1;i,j}\nonumber\\
 &\quad+\frac14\sum_{ijkl}v_{ij,kl}D_{2;ij,kl}
 \equiv\Tr[HD_2].
 \label{eq:contraction-lifting}
\end{align}
The rightmost notation is the RDM contraction used throughout; it does not
assert the existence of a positive $\rho_N$ for an RDM that is only $q$-positive.  Lifting is
not an inverse reconstruction of a density matrix.  Below, $H$ abbreviates
$H^{[N]}$ when the space is clear; $\pN$, $\SigmaN$, and $\coneN$
retain the explicit $[N]$ marker to emphasize their dependence on the
fixed-particle-number identities.

\subsection{Equivalence of full $(2,q)$-positivity and full $q$-positivity}
\label{sec:2q}

Mazziotti's $(2,q)$ conditions \cite{Mazziotti2012PRL,Mazziotti2012PRA} are
inequalities on the 2-RDM obtained from sums of squares of $q$-body
operators in which all three- and higher-body parts cancel exactly, so that
only a two-body operator remains.  Examples at increasing $q$ are as
follows.
\begin{itemize}
\item \emph{$q=2$.} The $D$, $Q$, and $G$ conditions are single squares
$C^\dagger C$ with $C$ a linear combination of $a_ia_j$, $a_i^\dagger
a_j^\dagger$, or $a_i^\dagger a_j$; no cancellation is needed.
\item \emph{$q=3$.} For $C=\sum b_{ijk}a_ia_ja_k$ the three-body terms of
$C^\dagger C$ and $CC^\dagger$ cancel in the sum, which gives $T_1$; the same
construction with $C=\sum b_{ijk}a_i^\dagger a_ja_k$ gives $T_2$, and adding
a one-body part gives $T_2'$ \cite{Zhao2004,Braams2007,Mazziotti2006}.
Mazziotti also derived lifted 2-positivity conditions, in which a $D$, $Q$,
or $G$ square is enlarged by an extra creation or annihilation operator and
the two choices are added so that the three-body parts cancel; these are
not implied by $D$, $Q$, $G$, $T_1$, and $T_2$
\cite[Eq.~(54)]{Mazziotti2012PRA}.
\item \emph{$q=4,5,6$.} Ref.~\cite{Mazziotti2012PRA} lists 3, 6, and 24
classes, respectively, of lifted and pure conditions, obtained with
tensor-factorized coefficients such as $b_{ij}d_ke_l$ in place of a general
$b_{ijkl}$.
\end{itemize}
By the bipolar theorem \cite{Kummer1967,Erdahl1978}, $N$-representable
2-RDMs are characterized by nonnegativity on such two-body operators, and
the constructive solution of Refs.~\cite{Mazziotti2012PRL,Mazziotti2012PRA}
generates all of them through the cancellation of higher-body operators:
those works present a systematic approach to complete
$N$-representability conditions, and the classes written out there are
examples generated by this complete constructive hierarchy.  Reference
\cite{Mazziotti2023} expresses the complete hierarchy at each level
through a single matrix equation [its Eq.~(7)], restated below as
\cref{eq:maz-eq}.  Following these works, full $(2,q)$-positivity means
the complete level-$q$ family.  Three relations will be used.  (i) Imposing
only the explicitly listed classes is weaker than full
$(2,q)$-positivity.  (ii) Every 2-RDM contracted from a full $q$-positive
$q$-RDM satisfies every $(2,q)$ condition, because a sum of squares has a
nonnegative expectation value.  (iii) Conversely, every 2-RDM satisfying
all $(2,q)$ conditions extends to a full $q$-positive $q$-RDM.
Statement~(iii) is the absence of a duality gap between the $q$-positivity
program and its polar.  It follows from Erdahl's fundamental theorem and
existence theorems for the lower-bound method
\cite[Theorems~10, 12, 13 and Corollary~14]{Erdahl2007}, cited in the
constructive solution \cite{Mazziotti2012PRA}, together with the
set inclusions of the constructive hierarchy
\cite{Mazziotti2012PRL,Mazziotti2012PRA,Mazziotti2023}.  We record it as
\cref{prop:projection} below in the fixed-$N$ form used here, as a
consequence of these established results rather than a new one.

To make this precise, let $\SigmaN$ be the cone generated by all squares
$C^\dagger C$ with $C\in\cW_{\le q}$, represented in the truncated CAR
quotient; by \cref{eq:fixed-N-lifting} the generators may be taken at exact
length $q$.  The cone of $(2,q)$ conditions is
\begin{equation}
 \coneN
 =
 \left\{
 B_2\in\cV_2:
 \pNadj B_2\in\SigmaN
 \right\},
 \label{eq:cancellationcone}
\end{equation}
where $\pNadj$ is the adjoint of contraction: it embeds the two-body
operator into the retained operator space, with fixed-$N$ identities
understood.  This is the truncated-space counterpart of the lifting in
\cref{eq:operator-lifting}, not a map that reconstructs an $N$-RDM.
In a $q$-particle tensor representation the embedding is written
$B_2\wedge I^{q-2}$ up to normalization; $q$ is the retained RDM order,
not the physical electron number $N$.  In particle--hole metric form the membership condition reads
\begin{equation}
 \pNadj B_2
 =
 \sum_{\alpha=0}^{q}\cM_{q,\alpha}^*(X_\alpha)+L^*z,
 \qquad X_\alpha\succeq0,
 \label{eq:maz-eq}
\end{equation}
with $z$ free.  Here $L$ collects the homogeneous linear identities $Ly=0$
of unnormalized moment vectors $y$, such as $y(\widehat N)-N\,y(\id)=0$, so
that $L^*z$ is an operator that vanishes on $\cH_N$; it is absent if these
identities are factored out from the start (Appendix~\ref{app:raw}).
Equation \eqref{eq:maz-eq} restates, in the particle--hole notation used
here, the finite-dimensional cancellation equations already given in
Ref.~\onlinecite{Mazziotti2023} [its Eq.~(7)] for the complete hierarchy;
explicit linear cancellation equations without particle-number
conservation are given in Ref.~\onlinecite{Mazziotti2026PRL}.  It differs
from Eq.~(7) of Ref.~\onlinecite{Mazziotti2023} only in bookkeeping: the
fixed-$N$ identities appear explicitly as $L^*z$, whereas there they are
built into the embedding $B_2\wedge I^{q-2}$ of the $q$-particle
representation.
In the weighted sum-of-squares form of
Refs.~\cite{RubinLowDePrince2026,RubinDePrince2026}, $L^*z$ contains the
particle-number multiplier terms $f(\widehat N-N)+(\widehat N-N)f^\dagger$;
the higher-body parts of the squares and of these terms need only cancel
together, and zero- and one-body multipliers $f$ already give the same
optimal bound \cite{RubinDePrince2026}.
The polar moment cone, its normalized slice, and the polar energy of a
two-body Hamiltonian $H$ are
\begin{align*}
 \cP_{2,q}^{\mathrm{polar},[N]}&=\coneN^*,\\
 \cF_{2,q}^{\mathrm{polar},[N]}
 &=\left\{D_2\in\cP_{2,q}^{\mathrm{polar},[N]}:D_2(\id)=1\right\}.
\end{align*}
\begin{equation}
 E_{(2,q)}(H)
 =\min_{D_2\in\cF_{2,q}^{\mathrm{polar},[N]}}
 \Tr\!\left[HD_2\right].
 \label{eq:polar-energy-definition}
\end{equation}
Writing $\cF_{2,q}^{\mathrm{listed},[N]}$ for the 2-RDMs that satisfy only the
listed classes, statements (i) and (ii) read
\begin{equation}
 \pN\!\left(\cF_q^{\mathrm{full},[N]}\right)
 \subseteq
 \cF_{2,q}^{\mathrm{polar},[N]}
 \subseteq
 \cF_{2,q}^{\mathrm{listed},[N]} ,
 \label{eq:polar-inclusions}
\end{equation}
so that $E^{\mathrm{listed}}_{(2,q)}(H)\le E_{(2,q)}(H)\le E_q(H)$; a
stronger condition admits fewer 2-RDMs and gives a higher bound.
Statement~(iii) turns the first inclusion into an equality, that is, it
excludes a duality gap.

\begin{proposition}[Projection--polar equivalence; Erdahl \cite{Erdahl2007}]
\label{prop:projection}
In a finite orbital basis, after quotienting CAR and fixed-$N$ operator
identities and removing structural null directions,
\begin{equation}
 \boxed{
 \cF_{2,q}^{\mathrm{polar},[N]}
 =
 \pN\!\left(\cF_q^{\mathrm{full},[N]}\right).
 }
 \label{eq:projection}
\end{equation}
Consequently, for every two-body Hamiltonian $H$,
\begin{equation}
 E_{(2,q)}(H)
 =
 E_q(H)
 :=\min_{D_2\in\pN(\cF_q^{\mathrm{full},[N]})}\Tr\!\left[HD_2\right].
 \label{eq:energy-equality}
\end{equation}
\end{proposition}

\begin{proof}[Outline of the proof]
The argument is that of Erdahl's Theorems 12 and 13 \cite{Erdahl2007},
written for the fixed-$N$ quotient; we include it to show where the
fixed-$N$ identities enter.  Contracting any full $q$-positive moment assignment gives two-body data
satisfying every $(2,q)$ condition: each condition is a sum of nonnegative
squares.  For the converse, the normalized full $q$-positive set is compact,
because its positive semidefinite RDM blocks have fixed traces.  Its
contracted two-body set is therefore compact and convex.  If some
normalized $D_2$ lies outside it, a two-body Hamiltonian separates $D_2$
from that set.  Subtracting the separating energy times the identity gives
a two-body operator $B_2$ whose expectation is negative at $D_2$ but
nonnegative for every full $q$-positive assignment.  The cone of sums of
$q$-level squares is closed in this fixed-$N$ setting; separation of that
cone then shows that $B_2$ is itself a $(2,q)$ condition.  Thus $D_2$
violates one of the assumed conditions, a contradiction.
Appendix~\ref{app:polar} supplies the boundedness and closedness arguments
and the two separation steps.  Closedness is not an automatic property of
semidefinite projections; an example of failure outside this setting is
given in the Supplementary Material.
\end{proof}

The proposition is Erdahl's duality theorem specialized to the fixed-$N$
setting, not a new representability condition.  The only
setting-specific hypothesis, that no nonzero positive semidefinite Gram
matrix represents the zero operator (Erdahl's condition
$\mathcal S\cap\mathcal P^k=0$ \cite[Theorem~13]{Erdahl2007}), is checked
in Appendix~\ref{app:polar} with the completely mixed fixed-$N$ state.  Equal optimal values also say nothing about computational cost:
in dense arithmetic both formulations carry blocks of linear dimension
$O(r^q)$, but the polar form may permit low-rank or otherwise structured
multipliers \cite{Mazziotti2023}, and such structure is where a
computational advantage would have to come from.

\section{A solvable counting model}
\label{sec:counting}

Let $r=2m$, divide the orbitals into disjoint sets
\[
 A=\{1,\ldots,m\},\qquad B=\{m+1,\ldots,2m\},
\]
and fix $N=m$; the case $N<m$ is treated in Sec.~\ref{sec:unbalanced}.  This is a half-filled system of two degenerate shells of
equal size: $A$ and $B$ carry no energy scale of their own, and the only
structure is which shell each fermion occupies.  Define
\begin{equation}
 K=\sum_{i\in A}n_i,\qquad n_i=a_i^\dagger a_i.
 \label{eq:counting}
\end{equation}
$K$ counts the fermions in shell $A$.  It commutes with the total number
operator and with every permutation of orbitals within $A$ and within $B$;
this symmetry is what makes the relaxation exactly solvable.  The operator
should be compared with the total number operator studied on Fock space in
Refs.~\cite{Erdahl1978,ErdahlGrudzinski1978}: there the particle number
itself is the variable, whereas here the total $N$ is fixed and $K$ probes a
subsystem occupation.  Every determinant contains an integer number $k$ of
particles in $A$, so
\begin{equation}
 \spec(K|_{\cH_m})=\{0,1,\ldots,m\}.
 \label{eq:exact-spectrum}
\end{equation}
For $\lambda\in\R$ consider
\begin{equation}
 H_\lambda=(K-\lambda\id)^2.
 \label{eq:dispersion}
\end{equation}
Since $K$ is one-body, $H_\lambda$ is at most two-body.  The exact ground
energy is
\begin{equation}
 E_{\mathrm{exact}}(\lambda)
 =
 \min_{k\in\{0,\ldots,m\}}(k-\lambda)^2,
 \label{eq:exact-energy}
\end{equation}
and its zero set is exactly the integer spectrum.  This is the dispersion
test in the sense of Refs.~\cite{ErdahlGrudzinski1978,Nakata2006}: as a
function of $\lambda$, the exact minimum is the squared distance from
$\lambda$ to the nearest integer, and it vanishes precisely at the
eigenvalues of $K$.

Let
\begin{equation}
 \begin{aligned}
 E_q(\lambda)
 &=
 \min_{D_2\in\pN(\cF_q^{\mathrm{full},[N]})}\Tr\!\left[H_\lambda D_2\right],\\
 Z_q(K)&=\{\lambda:E_q(\lambda)=0\}.
 \end{aligned}
 \label{eq:relaxed-zero}
\end{equation}
By \cref{prop:projection},
\begin{equation}
 E_{(2,q)}(\lambda)
 :=E_{(2,q)}(H_\lambda)
 =E_q(\lambda).
 \label{eq:equal-counting-energies}
\end{equation}
Thus $Z_q(K)$ is the approximate spectrum of $K$ relative to level $q$, in
the terminology of Ref.~\cite{Erdahl1978}, and the same set is obtained from
the primal and the polar formulation.  Two elementary bounds frame the
problem.  First, every physical eigenstate of $K$ defines a feasible
functional, so $E_q(\lambda)\le E_{\mathrm{exact}}(\lambda)$ and
$\{0,1,\ldots,m\}\subseteq Z_q(K)$: the relaxation can only add spurious
zeros, never remove physical ones.  Second,
$H_\lambda=C^\dagger C$ with $C=K-\lambda\id\in\cW_{\le2}$, so the
level-two condition alone, a square of a one-body operator, gives $E_q(\lambda)\ge0$ for every
$q\ge2$.  The relaxed dispersion is therefore never negative, and the only
possible failure of the spectral-zero criterion is that it reaches zero at
a noninteger $\lambda$; the question is exactly where this happens.
(Positive dispersion values can also be underestimated, but that does not
affect the zero set; see Sec.~\ref{sec:energy}.)  (If one prefers rows with exactly
two creation and annihilation operators, use $\id=\widehat N/N$ on the fixed-$N$ sector;
then $C=K-(\lambda/N)\widehat N$ is an ordinary one-body operator.)

The two directions of \cref{prop:projection} enter differently.  The
elementary inclusion
\[
 \pN(\cF_q^{\mathrm{full},[N]})\subseteq
 \cF_{2,q}^{\mathrm{polar},[N]}
\]
gives $E_{(2,q)}(\lambda)\le E_q(\lambda)$, and
$H_\lambda=C^\dagger C$ is itself a $(2,q)$ condition, so
$E_{(2,q)}(\lambda)\ge0$.  Hence
\[
 0\le E_{(2,q)}(\lambda)\le E_q(\lambda),
\]
and every relaxed zero of full $q$-positivity is a relaxed zero of full
$(2,q)$-positivity.  This is all that is needed for the existence of false
zeros and for the linear lower bound on the hierarchy depth.  The converse
inclusion, which follows from Erdahl's duality theorem
\cite{Erdahl2007} (\cref{prop:projection}), is used only to
show that full $(2,q)$-positivity has no further zeros, that is, that the
two zero sets coincide exactly and that the recovery threshold is sharp
for $(2,q)$ as well.

\section{Interpolated reduced density matrices and the case $q=3$}
\label{sec:pseudo}

\subsection{Interpolated reduced density matrices with zero dispersion}

To prove that a noninteger $\lambda$ is a relaxed zero at level $q$, we
construct by hand interpolated reduced density matrices that pass every level-$q$
condition and give zero dispersion.  Write
${}^1\!D_{ij}=\langle a_i^\dagger a_j\rangle_\lambda$ and
${}^2\!D_{ij,kl}=\langle a_i^\dagger a_j^\dagger a_la_k\rangle_\lambda$, where
$\langle\cdot\rangle_\lambda$ denotes the interpolated expectation value to be
constructed on the fixed-$N$ space.  Let $\rho_k$ be the equal mixture of
determinants with $k$ electrons in $A$ and $m-k$ in $B$.  Define
\begin{align}
 \langle O\rangle_\lambda:=\widehat{\mathcal L}_\lambda(O)
 &=\sum_{k=0}^{m}\ell_k(\lambda)\Tr(\rho_k O),\nonumber\\
 \ell_k(\lambda)
 &=\prod_{\substack{j=0\\j\ne k}}^{m}\frac{\lambda-j}{k-j}.
 \label{eq:fixedN-interpolation}
\end{align}
This signed interpolation is normalized and obeys every fixed-$N$ operator
identity, since each $\rho_k$ has exactly $N=m$ electrons.  Its weights
need not be positive.  We set $\Tr[OD]=\langle O\rangle_\lambda$ to obtain
the interpolated RDMs.  In the $N$-electron space, $n_i^2=n_i$ and
$n_in_j=a_i^\dagger a_j^\dagger a_ja_i$ give
\begin{equation}
 \Tr\!\left[(K-\lambda)^2D_2\right]
 =
 (1-2\lambda)\sum_{i\in A}{}^1\!D_{ii}
 +\sum_{\substack{i,j\in A\\ i\ne j}}{}^2\!D_{ij,ij}
 +\lambda^2 .
 \label{eq:dispersion-trace}
\end{equation}
The goal is to make the right-hand side vanish at a noninteger $\lambda$
while all blocks, for $q=3$ the blocks $D_2,Q_2,G_2,D_3,E_3,F_3,Q_3$,
remain positive semidefinite.

For $m=5$, a physical state with two electrons uniformly distributed in
$A$ has mean shell occupation $2$, pair occupation $1/10$, and zero
dispersion.  An equal mixture of the $k=2$ and $k=3$ ensembles has mean
$5/2$, pair occupation $1/5$, and dispersion $1/4$.  The signed
interpolation at $\lambda=5/2$ instead has pair occupation $3/16$ and
zero dispersion.  No positive mixture of determinants can have a
noninteger mean with zero variance; the question is which positivity
levels detect this change.

The construction follows two rules.  First, all entries that change the
occupation of some orbital are set to zero.  Because $K$ is diagonal in the
occupation basis, no such coherences are needed; formally this is an
average over the phases of all orbitals, \cref{eq:twirl} in
Appendix~\ref{app:general}.  Every remaining entry is then the
expectation of a product of particle factors $n_i$ and hole factors
$1-n_j$.

Second, these expectations are copied from a genuine state and then
continued in the particle number.  If exactly $k$ of the $m$ orbitals of a
shell are occupied, uniformly at random, then
\begin{gather*}
 \langle n_i\rangle=\frac km,\qquad
 \langle n_in_j\rangle=\frac{k(k-1)}{m(m-1)},\\
 \langle n_i(1-n_j)\rangle=\frac{k(m-k)}{m(m-1)},\\
 \langle n_in_jn_l\rangle=\frac{k(k-1)(k-2)}{m(m-1)(m-2)},\\
 \langle n_in_j(1-n_l)\rangle=\frac{k(k-1)(m-k)}{m(m-1)(m-2)},
\end{gather*}
the last two being diagonal entries of $D_3$ and of $E_3$ (row
$a_l^\dagger a_ja_i$).  Each particle factor contributes the next of
$k,k-1,\ldots$, each hole factor the next of $m-k,m-k-1,\ldots$.  These are
polynomials in $k$; we replace $k$ by $\lambda$ in shell $A$ and by
$m-\lambda$ in shell $B$.  For $m=5$ and $\lambda=5/2$ the two three-orbital
values are $1/32$ and $5/32$, summing to the pair occupation $3/16$ as
probabilities would.  At $\lambda=3/2$, however,
$\langle n_in_jn_l\rangle_\lambda=-1/160$, a negative entry of $D_3$; this is
why level three rejects $\lambda=3/2$.  In general, with
$\fall{z}{s}=z(z-1)\cdots(z-s+1)$,
\begin{equation}
 \begin{aligned}
 &\Bigl\langle\prod_{i\in P}n_i\prod_{j\in H}(1-n_j)\Bigr\rangle_\lambda\\
 &\quad=
 \frac{\fall{\lambda}{p_A}\fall{m-\lambda}{h_A}}{\fall{m}{p_A+h_A}}\,
 \frac{\fall{m-\lambda}{p_B}\fall{\lambda}{h_B}}{\fall{m}{p_B+h_B}},
 \end{aligned}
 \label{eq:pseudo-values}
\end{equation}
where $P$ and $H$ are disjoint sets of occupied and empty orbitals,
$p_A=|P\cap A|$, $h_A=|H\cap A|$, and similarly for $B$.
This evaluation formula holds for $|P|+|H|\le m$: the corresponding
expectation in $\rho_k$ is a polynomial in $k$ of degree at most
$|P|+|H|$, so interpolation at $k=0,\ldots,m$ reproduces it exactly.
The product expression can be extended algebraically to higher degrees,
but that auxiliary extension need not satisfy fixed-$N$ identities and
is not the definition of our fixed-$N$ interpolated state.

Every entry of a level-$q$ block uses occupation degree at most $q\le m$.
The sufficient classical positivity argument also uses squares of degree
at most $q-1$, hence moments through $2q-2\le m$ whenever the central
interval is nonempty.  Thus the product formula covers all required
positivity tests, while fixed-$N$ consistency holds exactly by
\cref{eq:fixedN-interpolation}.

\subsection{Block structure at arbitrary order}
\label{sec:arbitrary-order}

The general block structure is the link between the occupation formula
and complete fermionic positivity.  After the phase average,
every positivity block splits into sub-blocks labeled by the set $U$ of
orbitals from which electrons are removed and the set $V$ of orbitals to
which electrons are added.  A sub-block has the rows $a_V^\dagger a_U\,n_P$,
where $n_P$ is a product of occupation numbers of further orbitals and
$|U|+|V|+2|P|\le q$; each occupation number $n_i=a_i^\dagger a_i$ costs two
operators.  Normal ordering gives this form, and rows with different
$(U,V)$ have zero mutual entries under the phase average.  A sub-block's
entries are occupation-pattern expectations
\[
 \Bigl\langle\prod_{i\in U}n_i\prod_{j\in V}(1-n_j)\,
 n_Pn_{P'}\Bigr\rangle_\lambda .
\]
If no further occupation numbers fit into the rows, the sub-block is a
single probability for a pattern of occupied and empty orbitals on at most
$q$ orbitals, a product of falling factorials of $\lambda$ and $m-\lambda$
of order at most $q$.  For noninteger $\lambda$ these are all nonnegative
exactly when
\[
 q-1\le\lambda\le m-q+1 .
\]
Otherwise the quadratic form of the sub-block is the interpolated expectation of
the square of an occupation-number polynomial of degree at most $q-1$, so
it is nonnegative whenever all occupation-correlation matrices of order
$q-1$, with rows that are products of at most $q-1$ occupation numbers, are
positive semidefinite.  For the interpolated state this is a known property: it
holds on the same interval $q-1\le\lambda\le m-q+1$
\cite{Grigoriev2001,Lee2016,Potechin2019} (\cref{lem:slice-positive}; a
self-contained proof is given in Appendix~\ref{app:slice}).  For $q=3$ the
two cases are the diagonal entries and the matrices of
Sec.~\ref{sec:q3-check} below.  The matrix conditions hold throughout
this interval and can hold beyond it; the occupation-pattern probabilities
supply the tighter restriction.  Appendix~\ref{app:general} gives the
formal sub-block lemma and the general sufficiency proof.

\subsection{Explicit RDM blocks}

Every entry of a retained positivity block is a value
$\langle C^\dagger C'\rangle_\lambda$ for two
row operators $C$ and $C'$ (products of creation and annihilation
operators), which vanishes unless $C$ and $C'$ change the occupation of
each orbital in the same way, and otherwise equals $\pm$ a value of
\cref{eq:pseudo-values}.  Explicitly, the 1-RDM and the two-body blocks are
\begin{align}
 {}^1\!D_{ij}&=\delta_{ij}\langle n_i\rangle_\lambda,\nonumber\\
 (D_2)_{ij,kl}
 &=\langle a_i^\dagger a_j^\dagger a_la_k\rangle_\lambda
 =(\delta_{ik}\delta_{jl}-\delta_{il}\delta_{jk})
 \langle n_in_j\rangle_\lambda,\nonumber\\
 (Q_2)_{ij,kl}
 &=\langle a_ia_ja_l^\dagger a_k^\dagger\rangle_\lambda\nonumber\\
 &=(\delta_{ik}\delta_{jl}-\delta_{il}\delta_{jk})
 \langle(1-n_i)(1-n_j)\rangle_\lambda,\nonumber\\
 (G_2)_{ij,kl}
 &=\langle a_i^\dagger a_ja_l^\dagger a_k\rangle_\lambda\nonumber\\
 &=\delta_{ij}\delta_{kl}\langle n_in_k\rangle_\lambda
 +\delta_{ik}\delta_{jl}(1-\delta_{ij})
 \langle n_i(1-n_j)\rangle_\lambda,
 \label{eq:pseudo-blocks}
\end{align}
for $i\ne j$ in $D_2$ and $Q_2$, with $\langle n_in_i\rangle_\lambda
=\langle n_i\rangle_\lambda$ in $G_2$.  At third order we
write $E_3$ and $F_3$ as the matrices $\langle C^\dagger C'\rangle$ for the
row operators $a_k^\dagger a_ja_i$ and $a_i^\dagger a_j^\dagger a_k$, and $\delta^{i'j'k'}_{ijk}$ for the
generalized Kronecker delta, equal to the sign of the permutation taking
$(i,j,k)$ to $(i',j',k')$ and zero otherwise.  For distinct $i,j,k$ and
distinct $i',j',k'$,
\begin{align}
 (D_3)_{ijk,i'j'k'}
 &=\langle a_i^\dagger a_j^\dagger a_k^\dagger a_{k'}a_{j'}a_{i'}\rangle_\lambda
 =\delta^{i'j'k'}_{ijk}\langle n_in_jn_k\rangle_\lambda,\nonumber\\
 (Q_3)_{ijk,i'j'k'}
 &=\langle a_ia_ja_ka_{k'}^\dagger a_{j'}^\dagger a_{i'}^\dagger\rangle_\lambda
 \nonumber\\
 &=\delta^{i'j'k'}_{ijk}\langle(1-n_i)(1-n_j)(1-n_k)\rangle_\lambda,
 \nonumber\\
 (E_3)_{ijk,i'j'k'}
 &=\bigl\langle(a_k^\dagger a_ja_i)^\dagger a_{k'}^\dagger a_{j'}a_{i'}
 \bigr\rangle_\lambda\nonumber\\
 &=\delta^{i'j'}_{ij}\delta_{kk'}\langle n_in_j(1-n_k)\rangle_\lambda,
 \nonumber\\
 (F_3)_{ijk,i'j'k'}
 &=\bigl\langle(a_i^\dagger a_j^\dagger a_k)^\dagger
 a_{i'}^\dagger a_{j'}^\dagger a_{k'}\bigr\rangle_\lambda\nonumber\\
 &=\delta^{i'j'}_{ij}\delta_{kk'}\langle n_k(1-n_i)(1-n_j)\rangle_\lambda .
 \label{eq:pseudo-blocks3}
\end{align}
If an index of an $E_3$ or $F_3$ row operator is repeated, the
anticommutation relations reduce it to $\pm n_pa_s$ or $\pm n_pa_s^\dagger$
with $p\ne s$.  These operators have no matrix elements with the
distinct-index ones, and among themselves
\begin{equation}
 \begin{aligned}
 \langle(n_pa_s)^\dagger n_{p'}a_{s'}\rangle_\lambda
 &=\delta_{ss'}\langle n_sn_pn_{p'}\rangle_\lambda,\\
 \langle(n_pa_s^\dagger)^\dagger n_{p'}a_{s'}^\dagger\rangle_\lambda
 &=\delta_{ss'}\langle(1-n_s)n_pn_{p'}\rangle_\lambda,
 \end{aligned}
 \label{eq:pseudo-blocks3-repeated}
\end{equation}
with $n_pn_p=n_p$ when $p=p'$.  For fixed $s$ these are full matrices in
$p,p'$.  In this sense the interpolated state is
diagonal: all of its data are the occupation probabilities
\cref{eq:pseudo-values}.  The only off-diagonal entries are those of the
number-operator part of $G_2$, the first term in \cref{eq:pseudo-blocks}, and
the repeated-index parts \cref{eq:pseudo-blocks3-repeated} of $E_3$ and
$F_3$; these are where genuine matrix positivity must be checked.

\Cref{tab:block-construction} collects the row operators and positivity
ranges verified in Sec.~\ref{sec:q3-check}.

\begin{table}
\caption{What must be checked for the interpolated state at $q=3$.  The entries
are given in \cref{eq:pseudo-blocks,eq:pseudo-blocks3,eq:pseudo-blocks3-repeated}.
``Weight'' means that the entries are values of \cref{eq:pseudo-values}
and must be nonnegative; ``matrix'' means that a matrix of such values
must be positive semidefinite.  The last column gives the noninteger $\lambda$ for
which the condition holds (Sec.~\ref{sec:q3-check}).}
\label{tab:block-construction}
\begin{ruledtabular}
\begin{tabular}{llll}
Block & Row operators $C$ & Check & Holds for \\
\hline
$D_2$ & $a_ia_j$ & weight & $1\le\lambda\le m-1$ \\
$Q_2$ & $a_i^\dagger a_j^\dagger$ & weight & $1\le\lambda\le m-1$ \\
$G_2$, $i\ne j$ & $a_i^\dagger a_j$ & weight & $0<\lambda<m$ \\
$G_2$, $i=j$ & $1$, $n_i=a_i^\dagger a_i$ & matrix & $0<\lambda<m$ \\
$D_3$, $Q_3$ & $a_ia_ja_k$, $a_i^\dagger a_j^\dagger a_k^\dagger$ & weight & $2\le\lambda\le m-2$ \\
$E_3$, $F_3$, distinct & $a_k^\dagger a_ja_i$, $a_i^\dagger a_j^\dagger a_k$ & weight & $1\le\lambda\le m-1$ \\
$E_3$, $F_3$, repeated & $n_pa_s$, $n_pa_s^\dagger$ & matrix & $1\le\lambda\le m-1$ \\
\end{tabular}
\end{ruledtabular}
\end{table}

\begin{table}
\caption{Entries of the interpolated state for $m=5$, $\lambda=5/2$
($r=10$, $N=5$, $q=3$).  Both shells carry the continued particle number $5/2$, so the
values do not depend on the shell.  The last row is a $q=4$ entry.}
\label{tab:pseudo-numbers}
\begin{ruledtabular}
\begin{tabular}{lc}
Entry & Value \\
\hline
$\langle n_i\rangle$ & $1/2$ \\
$\langle n_in_j\rangle$, same shell & $3/16$ \\
$\langle n_in_j\rangle$, $i\in A$, $j\in B$ & $1/4$ \\
$\langle n_i(1-n_j)\rangle$, same shell & $5/16$ \\
$\langle(1-n_i)(1-n_j)\rangle$, same shell & $3/16$ \\
$\langle n_in_jn_k\rangle$, same shell & $1/32$ \\
$\langle n_in_j(1-n_k)\rangle$, same shell & $5/32$ \\
$\langle n_i(1-n_j)(1-n_k)\rangle$, same shell & $5/32$ \\
$\langle(1-n_i)(1-n_j)(1-n_k)\rangle$, same shell & $1/32$ \\
$\langle n_in_jn_k\rangle$, two in $A$, one in $B$ & $3/32$ \\
\hline
$\langle n_in_jn_kn_l\rangle$, same shell ($q=4$) & $-1/128$ \\
\end{tabular}
\end{ruledtabular}
\end{table}

\Cref{tab:pseudo-numbers} gives the minimal example $m=5$, $\lambda=5/2$.
The resulting 2-RDM has no coherences in the occupation basis: its only
data are the occupations $\langle n_i\rangle=1/2$ and the pair occupations
$\langle n_in_j\rangle=3/16$ for orbitals in the same shell and $1/4$ for
orbitals in different shells, a classical pair distribution that no
probability distribution over determinants can produce.  Nevertheless
$\langle K\rangle=5\cdot\frac12=\frac52$ and
$\langle K^2\rangle=\frac52+20\cdot\frac3{16}=\frac{25}4=\lambda^2$, so the
dispersion vanishes.  Section~\ref{sec:q3-check} shows that this
2-RDM, together with the three-body entries of \cref{tab:pseudo-numbers},
passes every condition of level $q=3$, including $D$, $Q$, $G$, $T_1$,
$T_2$, and all other $(2,3)$ conditions.  At $q=4$ the construction needs
the four-electron probability in the last row, which is negative; $D_4$ then
fails, and the false zero disappears.

\subsection{Positivity at $q=2$ and $q=3$}
\label{sec:q3-positivity}

This subsection has two parts.  First we determine the zero sets at
$q=2$ and $q=3$ by adding the positivity blocks one at a time.  Then we
verify that the interpolated RDMs pass every block of level three on the
accepted interval.

\subsubsection{Zero sets: adding blocks one at a time}
\label{sec:zero-ladders}

As a baseline, consider $q=2$, that is, the $D_2$, $Q_2$, and $G_2$
(PQG) conditions.  For the counting model the zero set is
\begin{equation}
 Z_2(K)=\{0\}\cup[1,m-1]\cup\{m\}\qquad(m\ge2).
 \label{eq:Z2}
\end{equation}
In words: besides the physical endpoints $0$ and $m$, every value from $1$
to $m-1$, integer or not, is accepted.  For $m=5$ the exact eigenvalues
are $0,1,\ldots,5$, whereas PQG accepts $\{0\}\cup[1,4]\cup\{5\}$; thus
$\lambda=3/2$ and $5/2$ pass, but $\lambda=1/2$ and $9/2$ do not.  We derive
\cref{eq:Z2} by imposing the three blocks one at a time: first $G_2$
alone, then $D_2$, then $Q_2$.  Throughout, normalization, the CAR
relations, and the fixed-$N$ identities are kept; only the positive
semidefinite conditions are added step by step.

\emph{Step 1: $G_2$ alone.}  With only the $G_2$ condition, the zero set is
\begin{equation}
 \boxed{Z_G(K)=[0,m]} .
 \label{eq:ZG}
\end{equation}
Every value between the two physical endpoints is accepted, integer or
not.  Two facts give this.

First, $G_2$ excludes nothing inside $[0,m]$.  For $0\le\lambda\le m$ the
interpolated RDMs of \cref{eq:pseudo-values} satisfy $G_2\succeq0$: its
single-excitation entries $\langle n_j(1-n_i)\rangle_\lambda$ are products of
$\lambda/m$, $(m-\lambda)/m$, and $\lambda(m-\lambda)/[m(m-1)]$, all
nonnegative, and its number-operator part is positive semidefinite by
\cref{eq:G2-zero-form} below.  By \cref{eq:dispersion-trace} the
dispersion is zero.  These interpolated RDMs are not physical for
$0<\lambda<1$: the pair occupation
$\langle n_in_j\rangle_\lambda=\lambda(\lambda-1)/[m(m-1)]$ is negative.
$G_2$ cannot see this, because it contains the pair occupations only
inside the positive semidefinite number-operator block.

Second, $G_2$ excludes everything outside $[0,m]$.  The rows $1$ and $n_i$
of $G_2$ give the $2\times2$ principal minor
$\langle n_i\rangle-\langle n_i\rangle^2\ge0$, so
$0\le\langle n_i\rangle\le1$.  Moreover, $K-\mu$ is a combination of the
rows $1,n_1,\ldots,n_r$ for every real $\mu$, so
$\langle(K-\mu)^2\rangle\ge0$.  Expanding,
\[
 \langle(K-\mu)^2\rangle
 =\langle(K-\lambda)^2\rangle+2(\lambda-\mu)\langle K-\lambda\rangle
 +(\lambda-\mu)^2 .
\]
The first term is the dispersion, which is zero.  If
$\langle K-\lambda\rangle\ne0$, a $\mu$ close to $\lambda$ on the appropriate
side would make the right-hand side negative.  Hence
\begin{equation}
 \langle K\rangle=\lambda,\qquad \langle K^2\rangle=\lambda^2 ,
 \label{eq:Z2-moments}
\end{equation}
and $\lambda=\sum_{i\in A}\langle n_i\rangle$ lies in $[0,m]$.

\emph{Step 2: adding $D_2$.}  The diagonal elements of $D_2$ are pair
occupations $\langle n_in_j\rangle$, which must now be nonnegative.  In
shell $A$, $n_i^2=n_i$ gives
\[
 K(K-1)=\sum_{\substack{i,j\in A\\ i\ne j}}n_in_j ,
\]
so \cref{eq:Z2-moments} yields
\begin{equation}
 \lambda(\lambda-1)\ge0 ,
 \label{eq:Z2-D}
\end{equation}
which removes the gap $(0,1)$.  At $\lambda=1/2$ the pair occupations of
$A$ would have to sum to $-1/4$.  Shell $B$ contains $\widehat N-K=m-K$
electrons on the fixed-$N$ space, so its pair occupations sum to
$(m-K)(m-K-1)$ and give
\begin{equation}
 (m-\lambda)(m-\lambda-1)\ge0 ,
 \label{eq:Z2-Q}
\end{equation}
which removes the gap $(m-1,m)$.  At $m=5$, $\lambda=9/2$ the pairs of $B$
would have to sum to $-1/4$.  Adding $D_2$ to $G_2$ therefore already
gives \cref{eq:Z2}.

\emph{Step 3: adding $Q_2$.}  The diagonal elements of $Q_2$ are hole-pair
occupations $\langle(1-n_i)(1-n_j)\rangle$.  In shell $A$ there are $m-K$
holes and in shell $B$ there are $K$, so the hole pairs sum to
$(m-K)(m-K-1)$ and $K(K-1)$.  These are the same two inequalities
\cref{eq:Z2-D,eq:Z2-Q}, and $Q_2$ removes nothing new.  This redundancy is
special to the counting model: at half filling with two equal shells,
the particle--hole exchange maps electron pairs of one shell to hole
pairs of the other.  In general, $Q_2$ is independent of $D_2$ and $G_2$.

The three steps are summarized as follows:
\begin{center}
\begin{tabular}{cccc}
 $\lambda$ & $G_2$ & $G_2+D_2$ & $G_2+D_2+Q_2$\\
 \hline
 $0$ & accepted & accepted & accepted\\
 $(0,1)$ & accepted & rejected & rejected\\
 $[1,m-1]$ & accepted & accepted & accepted\\
 $(m-1,m)$ & accepted & rejected & rejected\\
 $m$ & accepted & accepted & accepted
\end{tabular}
\end{center}
The endpoints $\lambda=0$ and $m$ survive every step because there the
products in \cref{eq:Z2-D,eq:Z2-Q} are exactly zero; they are the
physical states with all electrons in $B$ or all in $A$.  On $[1,m-1]$ the
interpolated RDMs also pass $D_2$ and $Q_2$: the pair and hole-pair occupations
$\fall{\lambda}{2}/\fall{m}{2}$ and $\fall{m-\lambda}{2}/\fall{m}{2}$ are
nonnegative because $1\le\lambda\le m-1$.  Together with Step~1 this gives
the equality \cref{eq:Z2}.

Now go one level up, to full three-positivity (3POS), which adds the
blocks $D_3$, $E_3$, $F_3$, and $Q_3$.  The zero set becomes
\begin{equation}
 Z_3(K)=\{0,1\}\cup[2,m-2]\cup\{m-1,m\}\qquad(m\ge4).
 \label{eq:Z3}
\end{equation}
In words: each end loses one more unit of the false interval.  The
integers are still all accepted, and so are the nonintegers between $2$
and $m-2$, but now also the gaps $(1,2)$ and $(m-2,m-1)$ are rejected.
For $m=5$, 3POS rejects $\lambda=1.5$ and $3.5$, which PQG accepted, but
still accepts $\lambda=2.5$.  As for PQG, we add the new blocks one at a
time on top of the PQG conditions: first $E_3$, then $D_3$, then $F_3$ and
$Q_3$.

\emph{Step 1: adding $E_3$.}  The rows of $E_3$ remove one electron; they
include $a_i(K-\mu)$ for every orbital $i$ and real $\mu$, since on the
fixed-$N$ space $a_i$ itself is a combination of the rows $a_in_j$.
Hence
\[
 \langle n_i(K-\lambda)^2\rangle
 =\bigl\langle\bigl(a_i(K-\lambda)\bigr)^\dagger a_i(K-\lambda)\bigr\rangle
 \ge0
\]
for every orbital $i$.
Summing over $i\in A$ gives $\langle K(K-\lambda)^2\rangle\ge0$; summing over
$i\in B$, where $\sum_{i\in B}n_i=m-K$ on the fixed-$N$ space, gives
$\langle(m-K)(K-\lambda)^2\rangle\ge0$.  The two sums add up to
$m\langle(K-\lambda)^2\rangle=0$, so both vanish.  With
\cref{eq:Z2-moments}, $\langle K(K-\lambda)^2\rangle=0$ is
\begin{equation}
 \langle K^3\rangle=\lambda^3 .
 \label{eq:Z3-moment}
\end{equation}
$E_3$ thus fixes the third moment, but it rejects no new value: on
$[1,m-1]$ the interpolated RDMs pass $E_3$ (its distinct-index entries are
nonnegative and its repeated-index parts are positive semidefinite there,
as checked below).  After this step the zero set is still $Z_2(K)$.

\emph{Step 2: adding $D_3$.}  The diagonal elements of $D_3$ are triple
occupations $\langle n_in_jn_k\rangle$, which must now be nonnegative.  In
shell $A$,
\[
 K(K-1)(K-2)=\sum_{\substack{i,j,k\in A\\ \text{distinct}}}n_in_jn_k ,
\]
so \cref{eq:Z2-moments,eq:Z3-moment} give
\begin{equation}
 \lambda(\lambda-1)(\lambda-2)\ge0 ,
 \label{eq:Z3-D}
\end{equation}
which removes the gap $(1,2)$.  At $m=5$, $\lambda=3/2$ the triples of $A$
would have to sum to $-3/8$: each of the $60$ ordered triples carries
$-1/160$.  (On $(0,1)$ the product is positive; that gap was already
removed by $D_2$.)  The $m-K$ electrons of shell $B$ give in the same way
\begin{equation}
 (m-\lambda)(m-\lambda-1)(m-\lambda-2)\ge0 ,
 \label{eq:Z3-Q}
\end{equation}
which removes the gap $(m-2,m-1)$; at $m=5$, $\lambda=7/2$ it again equals
$-3/8$.  Adding $D_3$ therefore already gives \cref{eq:Z3}.

\emph{Step 3: adding $F_3$ and $Q_3$.}  These are the hole counterparts of
$E_3$ and $D_3$.  The rows $a_i^\dagger(K-\mu)$ of $F_3$ give
$\langle(1-n_i)(K-\lambda)^2\rangle=0$, again equivalent to
\cref{eq:Z3-moment}.  The hole triples of $Q_3$ sum to
$(m-K)(m-K-1)(m-K-2)$ in shell $A$ and to $K(K-1)(K-2)$ in shell $B$,
reproducing \cref{eq:Z3-D,eq:Z3-Q}.  As for $Q_2$, nothing new is removed,
for the same particle--hole reason.

The three steps are summarized as follows:
\begin{center}
\begin{tabular}{cccc}
 $\lambda$ & PQG$+E_3$ & $+D_3$ & $+F_3+Q_3$\\
 \hline
 $0,\,1$ & accepted & accepted & accepted\\
 $(1,2)$ & accepted & rejected & rejected\\
 $[2,m-2]$ & accepted & accepted & accepted\\
 $(m-2,m-1)$ & accepted & rejected & rejected\\
 $m-1,\,m$ & accepted & accepted & accepted
\end{tabular}
\end{center}
The gaps $(0,1)$ and $(m-1,m)$, removed by PQG, are omitted.  The
integers survive because they are physical; at $\lambda=1$, $2$, $m-2$,
$m-1$ the products in \cref{eq:Z3-D,eq:Z3-Q} are exactly zero or positive.
On $[2,m-2]$ the interpolated RDMs pass all four new blocks: the triple and
hole-triple occupations $\fall{\lambda}{3}/\fall{m}{3}$ and
$\fall{m-\lambda}{3}/\fall{m}{3}$ are nonnegative because
$2\le\lambda\le m-2$, the mixed triples are nonnegative already for
$1\le\lambda\le m-1$, and the remaining matrices are checked below.  The
new factors $\lambda-2$ and $m-\lambda-2$ are the first step of the
interval contraction described by \cref{thm:main}.

\subsubsection{Verification that the interpolated RDMs pass level three}
\label{sec:q3-check}

It remains to verify the accepted values: that on $[2,m-2]$ the
interpolated RDMs pass every block of level three.  Every condition can be
checked by inspecting the blocks
\cref{eq:pseudo-blocks,eq:pseudo-blocks3,eq:pseudo-blocks3-repeated}
directly.  They contain two kinds of pieces.

\paragraph*{Diagonal entries.}
$D_2$, $Q_2$, the single-excitation part of $G_2$, $D_3$, $Q_3$, and the
distinct-index parts of $E_3$ and $F_3$ are diagonal up to antisymmetry, so
they are positive semidefinite exactly when every entry is nonnegative.
Each entry is a value of \cref{eq:pseudo-values}, a product of falling
factorials of $\lambda$ and $m-\lambda$; their signs are those worked out in
the two ladders above.  For noninteger $\lambda$, all diagonal entries are
nonnegative exactly when $2\le\lambda\le m-2$.

\paragraph*{Number-operator part of $G_2$.}
Here the rows are $1,n_1,\ldots,n_r$, so the condition is
$\langle p^\dagger p\rangle_\lambda\ge0$ for $p=c_0+\sum_ic_in_i$.  Write
$\bar c_A$ and $\bar c_B$ for the averages of $c_i$ over each shell.  Using
$\langle n_i\rangle=\theta/M$ and
$\langle n_in_k\rangle=\theta(\theta-1)/[M(M-1)]$ within a shell of $M$
orbitals with continued particle number $\theta$, one finds
\begin{equation}
 \langle p^\dagger p\rangle_\lambda
 =\bigl|\langle p\rangle_\lambda\bigr|^2
 +\gamma\sum_{i\in A}|c_i-\bar c_A|^2
 +\gamma\sum_{i\in B}|c_i-\bar c_B|^2,
 \label{eq:G2-zero-form}
\end{equation}
with $\gamma=\theta(M-\theta)/[M(M-1)]=\lambda(m-\lambda)/[m(m-1)]$ for both
shells.  The first term comes from the parts of $p$ that are constant on
each shell, that is, from $1$, $K$, and $\widehat N-K$, which have zero
variance; the other terms come from the fluctuations within each shell.
This block is positive semidefinite for $0\le\lambda\le m$.  For
$0<\lambda<m$ one has $\gamma>0$, and its null vectors are exactly the
combinations of $K-\lambda$ and $\widehat N-N$; at the endpoints
$\lambda=0,m$ the shell fluctuations also vanish and the null space is
larger.  The first
null vector is the zero dispersion.  For example, at $m=5$, $\lambda=5/2$
one has $\gamma=5/16$.  The operator $p=n_1-n_2$ with $1,2\in A$ has
$\langle p\rangle_\lambda=0$ and
$\langle p^\dagger p\rangle_\lambda=\tfrac12+\tfrac12-2\cdot\tfrac3{16}
=\tfrac58=2\gamma$, a positive fluctuation, whereas $p=K-\lambda$ gives
zero.

\paragraph*{Repeated-index parts of $E_3$ and $F_3$.}
For fixed $s$, the rows $a_s$ and $n_pa_s$ give entries
$\langle n_sn_pn_{p'}\rangle_\lambda$.  Because
$\fall{\lambda}{1+t}/\fall{m}{1+t}
=(\lambda/m)\,\fall{\lambda-1}{t}/\fall{m-1}{t}$, requiring an electron in
$s\in A$ simply removes one orbital and one unit of continued particle number from shell $A$:
the matrix equals $\langle n_s\rangle_\lambda=\lambda/m$ times a
number-operator $G_2$ block of the same form \cref{eq:G2-zero-form}, now with
$(\theta,M)=(\lambda-1,m-1)$ in shell $A$.  It is therefore positive
semidefinite for $1\le\lambda\le m$.  The same holds for $s\in B$, for which
$1\le m-\lambda$, and for $F_3$, where a hole in $s$ removes one orbital but
no particle.  Together these require $1\le\lambda\le m-1$.  For example,
at $m=5$, $\lambda=5/2$, the rows $a_s$ and $n_pa_s$ with $s,p\in A$ give
the $2\times2$ block
\[
 \begin{pmatrix}
 \langle n_s\rangle & \langle n_sn_p\rangle\\
 \langle n_sn_p\rangle & \langle n_sn_p\rangle
 \end{pmatrix}_\lambda
 =
 \begin{pmatrix}
 1/2 & 3/16\\
 3/16 & 3/16
 \end{pmatrix},
\]
with positive determinant $15/256$.

\paragraph*{Result.}
The ranges are collected in \cref{tab:block-construction}.  The
matrix conditions in $G_2$, $E_3$, and $F_3$ are the weaker ones; the
binding conditions are the signs of the three-electron and three-hole
probabilities in $D_3$ and $Q_3$.  For noninteger $\lambda$, every block at
$q=3$ is positive semidefinite exactly when $2\le\lambda\le m-2$, which is the central
interval of \cref{thm:main} for $q=3$.  On this interval the interpolated
reduced density matrices are 3-positive RDMs: they pass all conditions of the level $q=3$, including
$D$, $Q$, $G$, $T_1$, $T_2$, and all other $(2,3)$ conditions, while
$\Tr[(K-\lambda)^2D_2]=0$.  That these pieces exhaust the conditions at
$q=3$ follows from the sub-block structure in Appendix~\ref{app:general}.

\section{Exact false-zero set}

\begin{theorem}[False excitation continuum]
\label{thm:main}
Let $2\le q\le m$.  For the counting operator \cref{eq:counting}, the zero
set of the full order-$q$ positivity relaxation is
\begin{equation}
 \boxed{
 \begin{aligned}
 Z_q(K)={}&\{0,1,\ldots,q-2\}\\
 &{}\cup [q-1,m-q+1]\\
 &{}\cup \{m-q+2,\ldots,m\}.
 \end{aligned}}
 \label{eq:zero-set}
\end{equation}
An interval whose lower endpoint exceeds its upper endpoint is empty.
By \cref{prop:projection}, the same zero set is obtained from full
$(2,q)$-positivity.
\end{theorem}

\subsection{Sufficiency}

Every integer $k\in\{0,\ldots,m\}$ belongs to $Z_q(K)$: take any physical
state supported on determinants with exactly $k$ occupied orbitals in $A$.

It remains to prove that every
\begin{equation}
 q-1\le\lambda\le m-q+1
 \label{eq:central-lambda}
\end{equation}
is a relaxed zero.  For $q=3$ this was shown in
Sec.~\ref{sec:q3-check}.  For general $q$,
\cref{lem:general-positivity} in Appendix~\ref{app:general} shows that
$\langle\cdot\rangle_\lambda\in\cF_q^{\mathrm{full},[N]}$ on the interval
\cref{eq:central-lambda}.

Finally,
\[
 \langle K\rangle_\lambda=\lambda,\qquad
 \langle \fall{K}{2}\rangle_\lambda
 =\fall{\lambda}{2},
\]
so
\[
 \langle K^2\rangle_\lambda=\lambda^2,\qquad
 \langle (K-\lambda)^2\rangle_\lambda=0.
\]
This proves the inclusion ``$\supseteq$'' in \cref{eq:zero-set}.

\subsection{Necessity}
\label{sec:necessity}

Suppose $\langle\cdot\rangle\in\cF_q^{\mathrm{full},[N]}$ and
$\langle(K-\lambda)^2\rangle=0$.  Put $\mu=\langle K\rangle$.  Since positivity through order
two implies $\langle (K-t)^2\rangle\ge0$ for every real $t$,
\[
 0=\langle (K-\lambda)^2\rangle
  =\langle (K-\mu)^2\rangle+(\mu-\lambda)^2.
\]
Hence
\begin{equation}
 \langle K\rangle=\lambda,\qquad \langle K^2\rangle=\lambda^2.
 \label{eq:first-two}
\end{equation}

Let $T$ be any set of at most $q-2$ orbitals and let
$\chi_\varepsilon$ run over the $2^{|T|}$ occupation patterns on $T$, the
projectors that prescribe each orbital of $T$ to be occupied or empty.
For each pattern there is a product $V_\varepsilon$ of creation and
annihilation operators of length
$|T|$ such that
$V_\varepsilon^\dagger V_\varepsilon=\chi_\varepsilon$.  Because
$|T|+2\le q$,
\[
 \left\langle \chi_\varepsilon(K-\lambda)^2\right\rangle
 =
 \left\langle (V_\varepsilon(K-\lambda))^\dagger
 V_\varepsilon(K-\lambda)\right\rangle\ge0.
\]
The patterns sum to the identity, whereas their expectations sum to
$\langle (K-\lambda)^2\rangle=0$.  Therefore every term vanishes.  In particular,
\begin{equation}
 \langle n_S(K-\lambda)^2\rangle=0
 \qquad(S\subseteq A,\ |S|\le q-2).
 \label{eq:localized-zero}
\end{equation}
Since $K^t$ is a linear combination of the square-free monomials $n_S$ with
$|S|\le t$, \cref{eq:localized-zero} implies
\begin{equation}
 \langle K^t(K-\lambda)^2\rangle=0,
 \qquad 0\le t\le q-2.
 \label{eq:moment-rec}
\end{equation}
Writing $M_t=\langle K^t\rangle$ gives
\[
 M_{t+2}-2\lambda M_{t+1}+\lambda^2M_t=0.
\]
Together with \cref{eq:first-two}, induction yields
\begin{equation}
 \langle K^s\rangle=\lambda^s,\qquad 0\le s\le q.
 \label{eq:powers}
\end{equation}
Consequently,
\begin{equation}
 \begin{aligned}
 \langle \fall{K}{s}\rangle&=\fall{\lambda}{s},\\
 \langle \fall{m-K}{s}\rangle&=\fall{m-\lambda}{s},
 \qquad 1\le s\le q.
 \end{aligned}
 \label{eq:factorial}
\end{equation}

Now
\[
 \fall{K}{s}
 =
 s!\sum_{\substack{S\subseteq A\\|S|=s}}n_S,
 \qquad
 \fall{m-K}{s}
 =
 s!\sum_{\substack{S\subseteq A\\|S|=s}}
 \prod_{i\in S}(1-n_i).
\]
Each summand is the square of a product of $s$ annihilation (or creation)
operators, so full positivity gives
\begin{equation}
 \fall{\lambda}{s}\ge0,\qquad
 \fall{m-\lambda}{s}\ge0,
 \qquad 1\le s\le q.
 \label{eq:falling-positive}
\end{equation}

If $0\le x<q-1$ is not an integer, put $j=\lfloor x\rfloor$.
Then $j\le q-2$ and
$\fall{x}{j+2}<0$, because precisely its last factor is negative.
Therefore
\begin{equation}
 \fall{x}{s}\ge0\ (1\le s\le q)
 \Longleftrightarrow
 x\in\{0,\ldots,q-2\}\cup[q-1,\infty).
 \label{eq:falling-character}
\end{equation}
Applying \cref{eq:falling-character} to both $\lambda$ and $m-\lambda$
gives exactly the right-hand side of \cref{eq:zero-set}.  This proves the
opposite inclusion and completes the proof of \cref{thm:main}.

\subsection{Consequences}

The general zero-set formula immediately gives a counterexample at every
fixed level of the hierarchy.

\begin{corollary}[An explicit counterexample at every level]
\label{cor:every-q}
For every integer $q\ge2$, set
\begin{equation}
 N=m=2q-1,\qquad r=2m=4q-2,\qquad
 \lambda_*=q-\frac12.
 \label{eq:every-q-instance}
\end{equation}
Then
\begin{equation}
 E_{\mathrm{exact}}(\lambda_*)=\frac14,
 \qquad
 E_q(\lambda_*)=E_{(2,q)}(\lambda_*)=0.
 \label{eq:every-q-gap}
\end{equation}
Moreover, the entire interval $[q-1,q]$ is a relaxed zero interval, whereas
only its two endpoints are physical eigenvalues.
\end{corollary}

\begin{proof}
For $m=2q-1$, the central interval in \cref{eq:zero-set} is $[q-1,q]$.
Its midpoint is the half-integer $q-\tfrac12$, whose squared distance from the
integer spectrum is $1/4$.  The vanishing of $E_{(2,q)}(\lambda_*)$
follows from $0\le E_{(2,q)}\le E_q$ (Sec.~\ref{sec:counting}).
\end{proof}

The first three instances beyond two-positivity, at orders $q=3,4,5$, are
listed in \cref{tab:explicit-examples}.  The notation is worth emphasizing:
$q$ is the positivity order, while $r$ is the number of spin orbitals.

\begin{table*}[t]
\caption{Smallest balanced members of the family with a false continuum for
$q=3,4,5$.  The interior of the displayed interval consists entirely of
nonphysical relaxed zeros.}
\label{tab:explicit-examples}
\begin{ruledtabular}
\begin{tabular}{rrrrrrr}
$q$ & $N$ & $r$ & $\lambda_*$ & relaxed interval & exact energy & relaxed energy \\
3 & 5 & 10 & $5/2$ & $[2,3]$ & $1/4$ & $0$ \\
4 & 7 & 14 & $7/2$ & $[3,4]$ & $1/4$ & $0$ \\
5 & 9 & 18 & $9/2$ & $[4,5]$ & $1/4$ & $0$ \\
\end{tabular}
\end{ruledtabular}
\end{table*}

\begin{figure}[t]
 \centering
 \includegraphics[width=.96\linewidth]{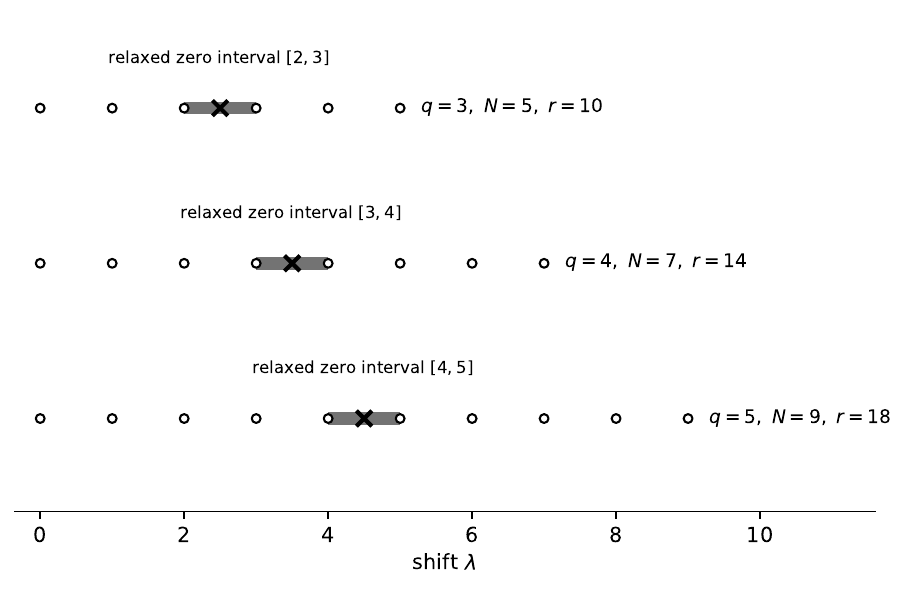}
 \caption{Explicit false continua for full $(2,3)$-, $(2,4)$-, and
 $(2,5)$-positivity.  Open circles mark the exact integer spectrum; the thick
 segment is the relaxed zero interval; and the cross marks the half-integer
 midpoint, where the exact dispersion is $1/4$ and the relaxed value is zero.}
 \label{fig:q345}
\end{figure}

\begin{corollary}[Minimal full-three-positivity counterexample]
\label{cor:q3}
For $q=3$, the smallest balanced counting model with a nontrivial false
continuum is $N=m=5$, that is, ten spin orbitals and five fermions.  Its zero
set is
\[
 Z_3(K)=\{0,1\}\cup[2,3]\cup\{4,5\}.
\]
At $\lambda=5/2$,
\[
 E_{\mathrm{exact}}(5/2)=1/4,
 \qquad
 E_3(5/2)=E_{(2,3)}(5/2)=0.
\]
The projected two-body moments therefore satisfy full $(2,3)$-positivity
while assigning zero dispersion to a value that is not an eigenvalue.
\end{corollary}

\begin{corollary}[Linear hierarchy depth]
\label{cor:linear}
The relaxed zero set equals the exact spectrum $\{0,1,\ldots,N\}$ for the
balanced counting family if and only if
\begin{equation}
 q\ge\left\lceil\frac{N+2}{2}\right\rceil.
 \label{eq:threshold}
\end{equation}
Hence exact recovery of the excitation zero set requires $q=\Omega(N)$.
\end{corollary}

\begin{proof}
For $N=m$, the only noninteger relaxed zeros are in
$[q-1,N-q+1]$.  This interval contains no noninteger point exactly when
$2q\ge N+2$.  Then the interval is empty or, for even $N$ at the threshold,
the single integer $q-1=N/2$; for example $N=4$, $q=3$ gives
$\{0,1\}\cup\{2\}\cup\{3,4\}$.  In either case \cref{eq:zero-set} is
exactly $\{0,1,\ldots,N\}$.
\end{proof}

\begin{figure}[t]
 \centering
 \includegraphics[width=.82\linewidth]{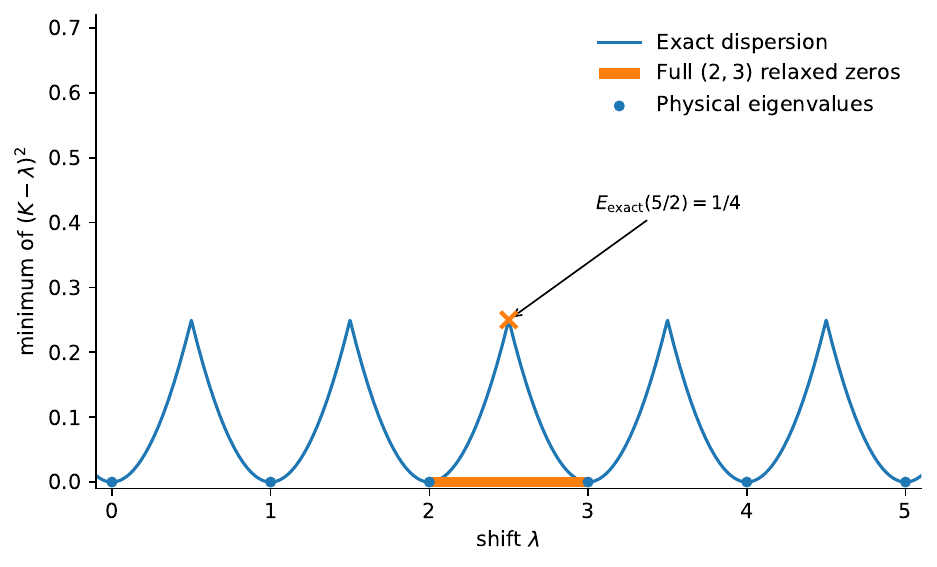}
 \caption{Minimal full-three-positivity failure.  The exact dispersion energy
 is the squared distance to the integer spectrum.  Full order-three
 positivity has zero energy at every point of $[2,3]$, including
 $\lambda=5/2$, where the exact value is $1/4$.}
 \label{fig:minimal}
\end{figure}

\begin{figure}[t]
 \centering
 \includegraphics[width=.90\linewidth]{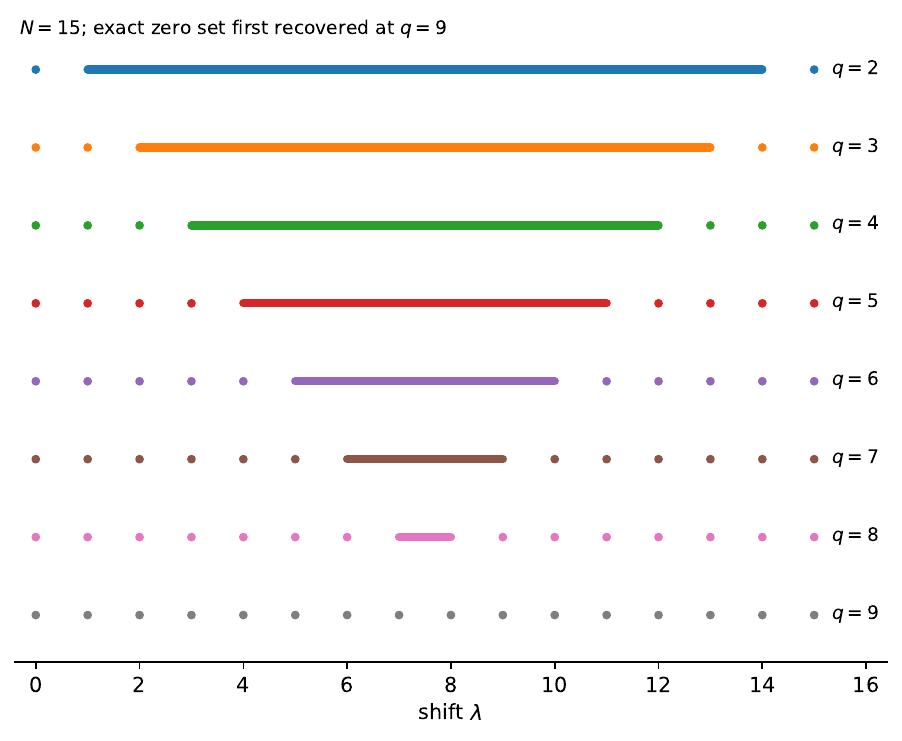}
 \caption{Zero sets from \cref{thm:main} for $N=15$.  Increasing $q$ recovers
 integer spectral points from the two edges, while a central false continuum
 remains until the linear threshold $q=9$ of \cref{eq:threshold}, where the
 zero set equals the spectrum.}
 \label{fig:hierarchy}
\end{figure}

\subsection{More orbitals than electrons}
\label{sec:unbalanced}

In practical calculations the number of orbitals per shell exceeds the
number of electrons.  Keep two shells $A$ and $B$ of $m$ spin orbitals each,
$r=2m$, but let $N\le m$.  The spectrum of $K$ is still
$\{0,1,\ldots,N\}$, and \cref{thm:main} holds with $m$ replaced by $N$:
for $2\le q\le N$,
\begin{equation}
 \begin{aligned}
 Z_q(K)={}&\{0,\ldots,q-2\}\cup[q-1,N-q+1]\\
 &{}\cup\{N-q+2,\ldots,N\}.
 \end{aligned}
 \label{eq:zero-set-unbalanced}
\end{equation}
Both the spurious continuum and the integer points outside it survive, and
the recovery threshold $q\ge\lceil(N+2)/2\rceil$ depends on $N$, not on $m$.

For necessity, the argument of Sec.~\ref{sec:necessity} applies unchanged,
except that the hole count $m-K$ of shell $A$ is replaced by the particle
count $N-K$ of shell $B$.  On $\cH_N$,
$\fall{N-K}{s}=s!\sum_{S\subseteq B,\,|S|=s}n_S$ is a sum of squares of
products of $s$ annihilation operators, and \cref{eq:powers} gives
$\langle\fall{N-K}{s}\rangle=\fall{N-\lambda}{s}$; hence
$\fall{N-\lambda}{s}\ge0$ for $1\le s\le q$.  Applying
\cref{eq:falling-character} to $\lambda$ and $N-\lambda$ gives the inclusion
``$\subseteq$'' in \cref{eq:zero-set-unbalanced}; the condition on
$m-\lambda$ is then implied.

For sufficiency, let $\rho_k$ be the equal mixture of determinants with $k$
electrons in $A$ and $N-k$ in $B$, and interpolate as in
\cref{eq:fixedN-interpolation} through the nodes $k=0,\ldots,N$.  The
evaluation formula \cref{eq:pseudo-values} becomes a product of shell $A$
with continued particle number $\lambda$ and shell $B$ with $N-\lambda$,
\[
 \frac{\fall{\lambda}{p_A}\fall{m-\lambda}{h_A}}{\fall{m}{p_A+h_A}}\,
 \frac{\fall{N-\lambda}{p_B}\fall{m-N+\lambda}{h_B}}{\fall{m}{p_B+h_B}},
\]
and holds for $|P|+|H|\le N$.  On the central interval $2q-2\le N\le m$, so
all required moments are covered, and
\[
 q-1\le\lambda\le m-q+1,\qquad q-1\le N-\lambda\le m-q+1.
\]
These are the hypotheses of \cref{lem:slice-positive} with $d=q-1$ in each
shell, and every pattern factor $\lambda-j$, $m-\lambda-j$, $N-\lambda-j$,
$m-N+\lambda-j$ with $j\le q-1$ is nonnegative, so the proof of
\cref{lem:general-positivity} applies verbatim.  For $q=3$ and
$(m,N)=(6,5)$, $(7,5)$, and $(8,6)$, the raw particle--hole blocks are
positive semidefinite with zero dispersion exactly on $[2,N-2]$, on a grid
of spacing $1/4$ in $\lambda$ (Supplementary Material).

\subsection{Numerical audit}

The proof is analytic.  The Supplementary Material reports independent
finite-matrix audits.  For $(r,N,q,\lambda)=(10,5,3,5/2)$, the raw
particle--hole blocks $D_1$ through $Q_3$, built from determinant-basis
creation and annihilation operators (dimensions $120$, $450$, $450$, and
$120$ at third order), are positive semidefinite to numerical precision,
while $D_3$ and $Q_3$ have the eigenvalue $-1/160$ at $\lambda=3/2$ and
$7/2$.  For the $q=3,4,5$ midpoints, all $35$, $70$, and $126$ classes of
sub-blocks (Appendix~\ref{app:general}) are positive semidefinite, whereas at
the noninteger points $\lambda=q-\tfrac32$ and $m-q+\tfrac32$ just outside
the interval a $q$-electron or $q$-hole probability is negative.  Integer
points outside the interval are physical eigenvalues and give positive
semidefinite matrices.  A reproduction package regenerates all audits,
tables, and figures.

\section{Scope of the result}
\label{sec:scope}

\subsection{Integrality is present in the two-body space}
\label{sec:integrality}

The lower bound concerns the full $q$-positivity hierarchy, not every method
that uses the 2-RDM.  Because the spectrum of $K$ is integer, the two-body
operators
\[
 B_j=(K-j\id)(K-(j+1)\id),\qquad j=0,\ldots,N-1,
\]
are nonnegative on every $N$-electron state.  On the interpolated state, however,
$\langle K\rangle_\lambda=\lambda$ and $\langle K^2\rangle_\lambda=\lambda^2$
give
\[
 \langle B_j\rangle_\lambda=(\lambda-j)(\lambda-j-1)<0
 \qquad(j<\lambda<j+1).
\]
Adding these $N$ linear two-body inequalities makes the entire dispersion
energy exact, not just its zero set.  To see this, let $k$ be a nearest
integer to $\lambda$ in $\{0,\ldots,m\}$ and put $\delta=\lambda-k$.
For $1\le k\le m-1$, $|\delta|\le1/2$ and
\begin{align*}
 (K-\lambda)^2-\delta^2
 &=\left(\tfrac12-\delta\right)B_{k-1}
 +\left(\tfrac12+\delta\right)B_k.
\end{align*}
Both coefficients are nonnegative.  At the endpoints the identities are
\begin{gather*}
 (K-\lambda)^2-\lambda^2=B_0+(1-2\lambda)K,
 \quad\lambda\le\tfrac12,\\
 \begin{aligned}
 (K-\lambda)^2-(m-\lambda)^2
 &=B_{m-1}\\
 &\quad+[1+2(\lambda-m)](m-K),
 \end{aligned}\\
 \lambda\ge m-\tfrac12.
\end{gather*}
Using the existing bounds $0\le\langle K\rangle\le m$ gives the same
lower bound as the exact dispersion; a physical eigenstate at $k$ attains
it.  The integrality information needed is therefore present
in the two-body space; the obstruction is that some of these inequalities
are not generated by the sums of squares of level $q$.  Every element of
$\coneN$ is nonnegative on $\langle\cdot\rangle_\lambda$, so choosing
$\lambda=j+\tfrac12$ in the false interval $[q-1,N-q+1]$ shows
\[
 B_j\notin\coneN\qquad(j=q-1,\ldots,N-q).
\]
The outermost inequalities are not obstructed: for example
\[
 B_0=K(K-1)=2\sum_{i<j\in A}(a_ja_i)^\dagger(a_ja_i)\in\cC_{2,2}^{[N]},
\]
and $B_{N-1}$ is the corresponding sum of hole-pair squares.  The theorem is
thus a statement about this hierarchy, not about the information content of
the 2-RDM or about the difficulty of the counting model.

The false continuum also shows where the logic of the dispersion test
breaks.  For a genuine state,
\[
 \Tr\rho(K-\lambda)^2=0
 \quad\Longrightarrow\quad
 (K-\lambda)\rho^{1/2}=0,
\]
so the state lies in the $\lambda$ eigenspace.  A truncated positive
functional enforces only finitely many of the consequences of this
implication.  It can make $K-\lambda$ a null vector of its order-two moment
matrix without being the restriction of any state.

\subsection{Relation to classical occupation-number relaxations}
\label{sec:relation}

Because the interpolated state is diagonal, every quantity in the positivity
conditions is a probability for a pattern of occupied and empty orbitals, or
a linear combination of such probabilities.  The problem is therefore a
problem about a classical distribution of occupation numbers
$n_i\in\{0,1\}$ with $\sum_{i\in A}n_i=\lambda$.  This is the setting of
known lower bounds for sum-of-squares relaxations of $0/1$ problems
\cite{Grigoriev2001,Lee2016,Potechin2019}, where it is called the knapsack
or Boolean setting.  In chemical language, the classical relaxation of
degree $d$ requires $\langle f^2\rangle\ge0$ for every polynomial $f$ of
degree at most $d$ in the occupation numbers, that is, positive
semidefiniteness of the occupation-correlation matrices whose rows are
products of at most $d$ occupation numbers.  The interpolated state of
Sec.~\ref{sec:pseudo} is the standard pseudo-expectation of this
literature, and the need
for a degree linear in the number of variables is already known for this
simple constraint.

The same diagonal sector appears in the constructive solution.  The
diagonal part of the 2-RDM is the moment data of Boolean quadratic
programming, whose classical $N$-representability conditions are
those of the Boole problem and of the cut and correlation polytopes
\cite{DezaLaurent1997}.  Reference~\onlinecite{Mazziotti2012PRA} discusses
this connection explicitly: when the expansion coefficients of its
$(2,4)$, $(2,5)$, and $(2,6)$ classes are orthogonal unit vectors, the
conditions reduce to classical diagonal $N$-representability conditions,
hypermetric inequalities for $q\le5$, while nonorthogonal coefficients
and reorderings of the creation and annihilation operators give genuinely
quantum conditions without classical counterpart.  The counting model
lies entirely in the classical sector, and the results below show that,
for its dispersion zero set, the quantum conditions at level $q$ cut off
nothing beyond a classical relaxation.  Equivalently, \cref{thm:main} is a
statement about Boolean quadratic programming with a cardinality
constraint: for the 0/1 problem of minimizing $(K-\lambda)^2$ subject to
$\sum_in_i=N$, full $q$-positivity and the classical relaxation
$\mathcal B_q^{[N]}$ defined below have the same zero set
\cref{eq:zero-set}.  We now state how the fermionic hierarchy relates to
these classical relaxations.

Let $\mathcal B_q^{[N]}$ denote the classical relaxation that imposes, for a
normalized functional on the diagonal operators of $\cH_N$,
(i) $\langle f^2\rangle\ge0$ for every occupation-number polynomial $f$ of
degree at most $q-1$, and (ii) nonnegativity of every occupation-pattern
probability
$\langle\prod_{i\in U}n_i\prod_{j\in V}(1-n_j)\rangle$ with
$|U|+|V|\le q$.  Only occupation moments through degree $2q-2$ are
needed.  On the nonempty central interval $2q-2\le m$, so all these
moments of the fixed-$N$ interpolation \cref{eq:fixedN-interpolation}
agree with the product evaluation formula of Sec.~\ref{sec:pseudo}.

\begin{proposition}[Classical form of the zero set]
\label{prop:classical}
For the counting model and $2\le q\le m$, the zero set $Z_q(K)$ of full
$q$-positivity equals the zero set of $\mathcal B_q^{[N]}$.
\end{proposition}

\begin{proof}
A functional satisfying $\mathcal B_q^{[N]}$, extended to all operators by the
phase average \cref{eq:twirl}, is full positive through order $q$: by
\cref{lem:car-sector}, every quadratic form of a positivity block is a sum
of terms $\langle\chi_{U,V}|p|^2\rangle=\langle|\chi_{U,V}p|^2\rangle$, where
$\chi_{U,V}$ is an occupation pattern and $p$ an occupation-number
polynomial, and the proof of \cref{lem:general-positivity} shows that each
term is controlled by (i) or (ii).  The dispersion operator is diagonal, so
the zero set of $\mathcal B_q^{[N]}$ is contained in $Z_q(K)$.  Conversely, the
fixed-$N$ functional \cref{eq:fixedN-interpolation} satisfies $\mathcal B_q^{[N]}$
on the central interval: its required moments equal the product formula,
whose positivity is established in \cref{lem:general-positivity}.
Integer values are physical, so the
right-hand side of \cref{eq:zero-set}, which equals $Z_q(K)$ by
\cref{thm:main}, is contained in the zero set of $\mathcal B_q^{[N]}$.
\end{proof}

The necessity part of \cref{thm:main} also uses only constraints of
$\mathcal B_q^{[N]}$: the localized dispersions
$\langle\chi_T(K-\lambda)^2\rangle\ge0$ with $|T|\le q-2$ are squares of
polynomials of degree at most $q-1$, and the remaining inputs are
probabilities of $q$ occupied or $q$ empty orbitals.

Thus the two relaxations have the same dispersion zero set for this model,
not equivalent conditions on a functional.  Indeed, the full three-positive
signed mixture in Sec.~\ref{sec:energy} gives
$\langle[(K-1)(K-2)]^2\rangle=-1/5$, and hence fails classical degree-two
positivity.  The degree $q-1$ in the sufficient condition follows from the
row length $|U|+|V|+2|P|$ of $a_V^\dagger a_Un_P$: each occupation number
costs a creation and an annihilation operator.

For $2(q-1)\le m$, classical degree-$(q-1)$ positivity of a shell
interpolated state holds at noninteger particle numbers in the wider range
$q-2<\lambda<m-q+2$ (Appendix~\ref{app:general}).  The additional
occupation-pattern probabilities restrict the noninteger values to the
central interval of \cref{thm:main}.  For example, at $m=5$, $q=3$,
$\lambda=3/2$, the degree-two occupation-correlation matrix is positive
semidefinite, but $\langle n_1n_2n_3\rangle=-1/160$ violates $D_3$.
The sufficiency proof uses both kinds of conditions.

The fermionic level $q$ does not reach the classical degree $q$.  Consider
the standard classical relaxation of degree $q$ for a single shell of $m$
orbitals: normalized functionals on polynomials in commuting variables
$x_i\in\{0,1\}$ with
\begin{align*}
 &\langle1\rangle=1,\\
 &\langle p^2\rangle\ge0\quad(\deg p\le q),\\
 &\Bigl\langle\Bigl(\sum_ix_i-\lambda\Bigr)h\Bigr\rangle=0\quad(\deg h\le2q-1).
\end{align*}
The polynomials $p$ have degree at most $q$, so $p^2$ and the moments
involved have degree at most $2q$.  Averaging over permutations of the
variables preserves these conditions, so a solution may be taken symmetric,
$\langle x_S\rangle=e_{|S|}$ with $e_0=1$, and the linear conditions with
$h=x_S$ give
\[
 (k-\lambda)e_k+(m-k)e_{k+1}=0\qquad(|S|=k),
\]
hence $e_k=\fall{\lambda}{k}/\fall{m}{k}$ for $k\le m$.  For $m=2q-1$ the
monomial $h=x_1\cdots x_m$ has the admissible degree $2q-1$, and since
$x_i^2=x_i$ the corresponding condition reads
\[
 0=(m-\lambda)e_m=(m-\lambda)\frac{\fall{\lambda}{m}}{m!},
\]
which fails for $\lambda=q-\tfrac12$.  The classical degree-$q$ relaxation is
therefore infeasible at the threshold for every $q$.  Infeasibility does
not rest on this last condition alone: if the linear conditions are imposed
only for $\deg h\le2q-2$, the recursion still fixes the only candidate,
\[
 M_{P,Q}=\frac{\fall{\lambda}{|P\cup Q|}}{\fall{m}{|P\cup Q|}},
 \qquad |P|,|Q|\le q,
\]
which is the single-shell interpolated state itself, and its smallest eigenvalue
is $-0.01428$ for $q=3$ ($m=5$, $26\times26$) and $-0.00416$ for $q=4$
($m=7$, $99\times99$).  At the same point the fermionic level $q$ still has
a false zero.  \Cref{tab:boolean-vs-fermion} summarizes the comparison.

\begin{table*}
\caption{Classical occupation-number relaxations and the fermionic
hierarchy for the counting model.  ``Classical degree $d$'' requires
$\langle f^2\rangle\ge0$ for every polynomial $f$ of degree at most $d$ in
the occupation numbers \cite{Grigoriev2001,Lee2016,Potechin2019}; a pattern
probability is $\langle\prod_{i\in U}n_i\prod_{j\in V}(1-n_j)\rangle$.}
\label{tab:boolean-vs-fermion}
\begin{ruledtabular}
\begin{tabular}{lll}
Relaxation & Constraints & Relation to full $q$-positivity \\
\hline
Classical degree $q-1$ & \parbox[t]{0.36\linewidth}{\raggedright $\langle f^2\rangle\ge0$ for $\deg f\le q-1$} & \parbox[t]{0.36\linewidth}{\raggedright implied by $\mathcal B_q^{[N]}$; alone not sufficient} \\
$\mathcal B_q^{[N]}$ & \parbox[t]{0.36\linewidth}{\raggedright degree $q-1$, and pattern probabilities on at most $q$ orbitals} & \parbox[t]{0.36\linewidth}{\raggedright same zero set (\cref{prop:classical})} \\
Full $q$-positivity & \parbox[t]{0.36\linewidth}{\raggedright all particle--hole blocks, rows of length at most $q$} & \parbox[t]{0.36\linewidth}{\raggedright ---} \\
Classical degree $q$ & \parbox[t]{0.36\linewidth}{\raggedright $\langle f^2\rangle\ge0$ for $\deg f\le q$} & \parbox[t]{0.36\linewidth}{\raggedright implies $\mathcal B_q^{[N]}$; infeasible at $m=2q-1$, $\lambda=q-\frac12$} \\
\end{tabular}
\end{ruledtabular}
\end{table*}

What is new, therefore, is not the interpolated state or its positivity, which are
known, but three things: the transfer of classical positivity and pattern
nonnegativity to all fermionic particle--hole blocks through the sub-block
structure; the necessity argument, which shows that the fermionic
constraints exclude every other value; and the exact zero set and sharp
threshold that follow.

\subsection{Zero set versus dispersion energy}
\label{sec:energy}

\Cref{thm:main} determines only the zero set, not the function
$E_q(\lambda)$ away from zero.  For a one-body operator $h$ and a fixed,
normalized RDM relaxation $\mathcal R$, the minimum dispersion is
\begin{equation}
 F_{\mathcal R}(\lambda)=\lambda^2+
 \min_{\mathcal L\in\mathcal R}
 \bigl[\mathcal L(h^2)-2\lambda\mathcal L(h)\bigr].
 \label{eq:dispersion-envelope}
\end{equation}
Each fixed feasible functional gives a parabola, but their lower envelope
need not consist of parabolic branches, because the minimizing RDM can vary
continuously with $\lambda$.  Exact minimization gives instead the lower
envelope of the finitely many parabolas centered at the physical
eigenvalues.

Recovery of the zero set does not imply that the dispersion energy is
exact.  Take $m=N=4$ and $q=3$, for which \cref{thm:main} gives
$Z_3(K)=\{0,1,2,3,4\}$, and let $\rho_k$ be the equal mixture of all
determinants with $K=k$.  The signed combination
\[
 \widetilde\rho=-\tfrac1{20}\rho_0+\tfrac35\rho_1+\tfrac9{20}\rho_2
\]
is not a density matrix, but it is diagonal and symmetric within each
shell, and its expectation values are full three-positive.  This was
verified in exact rational arithmetic: by \cref{lem:car-sector}, full
three-positivity reduces to 35 classes of sub-blocks of dimension at most
9, and a Schur-complement test shows that all of them are positive
semidefinite (Supplementary Material).  At $\lambda=3/2$,
\[
 \Tr\widetilde\rho\,(K-\tfrac32)^2=\tfrac3{20}<\tfrac14
 =E_{\mathrm{exact}}(\tfrac32),
\]
so $E_3(3/2)\le3/20$.  Throughout, ``recovery'' refers to the zero set
only.

This distinction is already visible in the Type-C example of
Nakata \emph{et al.} \cite{Nakata2006}, particularly Fig.~5 of its electronic
supplement.  The PQG calculation exhibits a zero continuum on $[-11,1]$
and nonparabolic shoulders, and the reported identification of ten out of
nineteen eigenvalues refers to isolated zeros, the remaining eigenvalues
being embedded in the continuum.  Those numerical profiles are consistent
with, but not quantitatively described by, the present theorem, which
treats a different Hamiltonian and the complete hierarchy.

\subsection{Small perturbations}

A generic splitting of the one-body levels destroys the exact permutation
symmetry.  A first quantitative statement is nevertheless immediate, at a
fixed system size.  Let $h_\varepsilon=K+\varepsilon V$ with $V$ a
Hermitian, particle-number conserving one-body operator, and let $\lambda$
lie in the central interval of \cref{thm:main}.  Since
$\langle(K-\lambda)^2\rangle_\lambda=0$ and the interpolated state is nonnegative on
squares of one-body operators, the Cauchy--Schwarz inequality removes the
cross terms, and
\[
 \langle(h_\varepsilon-\lambda)^2\rangle_\lambda
 =\varepsilon^2\langle V^2\rangle_\lambda .
\]
Hence the relaxed dispersion obeys
$0\le E_q^{(\varepsilon)}(\lambda)\le\varepsilon^2\langle
V^2\rangle_\lambda$ on the whole interval.  The exact dispersion, in
contrast, stays finite at half-integers: with $\|V\|$ the operator norm on
$\cH_N$, every eigenvalue of $h_\varepsilon$ lies within $|\varepsilon|\|V\|$
of an integer, so
\[
 E_{\mathrm{exact}}^{(\varepsilon)}(\lambda)\ge
 \Bigl(\tfrac12-|\varepsilon|\|V\|\Bigr)^2
 \qquad\bigl(\lambda\in\mathbb Z+\tfrac12,\ |\varepsilon|\|V\|<\tfrac12\bigr).
\]
A low-dispersion valley of depth $O(\varepsilon^2)$ therefore survives a
small splitting, whereas the exact dispersion remains close to $1/4$.
These bounds hold at fixed $N$; if $\|V\|$ grows with $N$, the admissible
$\varepsilon$ shrinks, and no robustness uniform in system size follows.
Whether exact zeros survive requires a separate analysis; the multilevel
examples of Ref.~\onlinecite{Nakata2006} show that splitting alone need not
eliminate a zero interval under PQG constraints.

\subsection{Density--density Hamiltonians}
\label{sec:density-density}

The dispersion operator is a simple density--density Hamiltonian.
Expanding,
\begin{equation}
 H_\lambda
 =
 \sum_{i\in A}(1-2\lambda)n_i
 +2\sum_{\substack{i<j\\i,j\in A}}n_in_j
 +\lambda^2\id .
 \label{eq:expand-H}
\end{equation}
On the fixed-$N$ sector, one-body and constant terms can be rewritten with
$n_in_j$ using
\begin{equation}
 n_i=\frac{1}{N-1}\sum_{j\ne i}n_in_j,\qquad
 \id=\frac{1}{\binom N2}\sum_{i<j}n_in_j,
 \label{eq:fixedN-rewrite}
\end{equation}
so $H_\lambda$ lies in the class $\sum_{ij}e_{ij}n_in_j$ emphasized in
diagonal $N$-representability.  For the midpoint $\lambda=N/2$, with the
pair counts
$P_{AA}=\sum_{i<j\in A}n_in_j$, $P_{BB}=\sum_{i<j\in B}n_in_j$, and
$P_{AB}=\sum_{i\in A,j\in B}n_in_j$, one has
\begin{equation}
 \begin{aligned}
 \left(K-\frac N2\right)^2
 ={}&\frac{N}{2(N-1)}(P_{AA}+P_{BB})\\
 &{}-\frac{N-2}{2(N-1)}P_{AB},
 \end{aligned}
 \label{eq:midpoint-density-density}
\end{equation}
so the $q=3,4,5$ examples of \cref{tab:explicit-examples} are
fixed-magnetization Ising Hamiltonians with a completely explicit
spectrum.  More general fixed-cardinality quadratic $0/1$ objectives include
Max-Cut and Ising ground-state problems, which are NP-hard
\cite{Garey1976,Barahona1982}, and weighted Max-Bisection, which remains
NP-hard on restricted graph classes \cite{DiazKaminski2007}.
Since full $q$-positivity has polynomial size $r^{O(q)}$ at fixed $q$,
exactness of one fixed level for every density--density instance would,
subject to the usual promise-gap and bit-complexity qualifications, solve an
NP-hard family in polynomial time.  That conditional observation shows only
that no fixed level can be exact for all such Hamiltonians.  \Cref{cor:linear}
is different: it proves linear growth, unconditionally, for an instance whose
exact spectrum and optimum are elementary.

\subsection{Open problems}

\Cref{cor:linear} gives the sharp threshold for this family, asymptotically
$N/2$; it does not assert that $q=N$ is necessary for all density--density
Hamiltonians.  Since $q=N$ is exact for every two-body Hamiltonian
(\cref{prop:projection} and Ref.~\onlinecite{RubinDePrince2026}), the level needed in the worst case to recover
the dispersion zero set lies between about $N/2$ and $N$; a stronger
worst-case lower bound would require a different explicit family.  Extending the construction to split levels would require a
weighted subset-sum interpolated state instead of the symmetric one used here.
More generally, quantifying how much the off-diagonal particle--hole blocks
add to the classical occupation-number relaxation of
Sec.~\ref{sec:relation}, and determining the minimal rank of complete polar
certificates, would clarify when a polar implementation can offer more
than a primal--dual exchange of $q$-body objects.

\section{Conclusion}

Complete full $(2,q)$-positivity does not, at fixed order, recover the
many-particle spectrum of even an elementary one-body counting operator.  In
a balanced $N$-particle system it accepts an exactly characterized interval
of noninteger values as zero-dispersion eigenvalues, and removing that
interval requires
\[
 q\ge\left\lceil\frac{N+2}{2}\right\rceil .
\]
The interpolated RDMs adapt a known classical construction to all particle--hole
blocks at arbitrary order.  Together with the necessity proof, this gives
the complete zero set and a sharp recovery threshold.  For this counting
model the same zero set follows from classical positivity of squares of
occupation polynomials of degree at most $q-1$ \emph{together with}
nonnegative occupation patterns on at most $q$ orbitals; square positivity
alone is insufficient.  The verified primal--polar equivalence transfers
the result to complete $(2,q)$-positivity.

\section*{Supplementary Material}
See the supplementary material for the exact algebraic and finite-matrix
audits, the exact positivity check of the signed mixture, an example of a
duality gap without strict feasibility, and the reproduction package.

\section*{Author Declarations}

\subsection*{Conflict of Interest}
The author has no conflicts to disclose.

\subsection*{Author Contributions}
Maho Nakata: Conceptualization (lead); Formal analysis (lead); Investigation
(lead); Methodology (lead); Software (lead); Validation (lead); Visualization
(lead); Writing--original draft (lead); Writing--review and editing (lead).

\subsection*{Use of Artificial Intelligence}
GPT-6 Astra (OpenAI) and Claude Opus 5.5 (Anthropic) were used to
assist with literature search and English-language drafting, and as aids in
the course of the research.  The author takes full responsibility for the
content of the manuscript.

\section*{Data Availability}
All data supporting the findings of this study are contained within the
article and its supplementary material.  The supplementary material includes
a reproduction package with the verification scripts, their logs, the
figure scripts, the \LaTeX{} sources, and a Makefile; \texttt{make verify},
\texttt{make figures}, and \texttt{make pdf} regenerate all reported audits,
tables, figures, and documents.

\appendix

\section{Guide to notation}
\label{app:notation}

Tables~\ref{tab:notation-system}--\ref{tab:notation-cones} collect the main
symbols, grouped by their physical or mathematical role.  In particular,
$q$ counts creation and annihilation operators in a positivity-test row,
whereas $d$ counts occupation numbers in a polynomial.  The model has
$N=m$ and $r=2m$; the symbols $M$ and $\theta$ refer instead to an auxiliary
single shell.  Interpolated expectations specify RDM entries but need not come
from a positive many-electron density operator.

\begin{table*}[t]
\caption{System, reduced density matrices, and hierarchy level.}
\label{tab:notation-system}
\begin{ruledtabular}
\begin{tabular}{ll}
\parbox[t]{0.22\textwidth}{\raggedright Symbol} & \parbox[t]{0.72\textwidth}{\raggedright Meaning} \\
\hline
\parbox[t]{0.22\textwidth}{\raggedright $r$, $N$} & \parbox[t]{0.72\textwidth}{\raggedright Number of spin orbitals and fixed total number of electrons.} \\
\parbox[t]{0.22\textwidth}{\raggedright $a_i^\dagger$, $a_i$, $n_i$} & \parbox[t]{0.72\textwidth}{\raggedright Creation, annihilation, and occupation-number operators; $n_i=a_i^\dagger a_i$.} \\
\parbox[t]{0.22\textwidth}{\raggedright $\widehat N$, $\id$} & \parbox[t]{0.72\textwidth}{\raggedright Total number operator $\sum_i n_i$ and identity operator.} \\
\parbox[t]{0.22\textwidth}{\raggedright $\rho_N$, $H^{[N]}$} & \parbox[t]{0.72\textwidth}{\raggedright Many-electron density operator and lifted Hamiltonian on $\cH_N$; $[N]$ labels the space, not interaction order.} \\
\parbox[t]{0.22\textwidth}{\raggedright $\cH_N$, $P_N$, $\mathcal F$} & \parbox[t]{0.72\textwidth}{\raggedright $N$-electron Hilbert space, its projector, and fermionic Fock space.} \\
\parbox[t]{0.22\textwidth}{\raggedright $C$, $\cW_{\le q}$} & \parbox[t]{0.72\textwidth}{\raggedright Test row operator and the span of rows $CP_N$ with at most $q$ creation/annihilation factors; rows may change particle number.} \\
\parbox[t]{0.22\textwidth}{\raggedright $q$} & \parbox[t]{0.72\textwidth}{\raggedright Fermionic positivity level: $\langle C^\dagger C\rangle\ge0$ for every $C\in\cW_{\le q}$.} \\
\parbox[t]{0.22\textwidth}{\raggedright $d$} & \parbox[t]{0.72\textwidth}{\raggedright Degree of an occupation-number polynomial: the largest number of factors $n_i$ in a monomial.  Each $n_i$ uses two fermionic factors.} \\
\parbox[t]{0.22\textwidth}{\raggedright $\cL$, $\langle O\rangle$} & \parbox[t]{0.72\textwidth}{\raggedright Linear assignment of retained operator expectations, with $\langle O\rangle=\cL(O)$; not assumed to be a physical state.} \\
\parbox[t]{0.22\textwidth}{\raggedright $D_2$} & \parbox[t]{0.72\textwidth}{\raggedright The two-particle RDM block; in $\Tr[HD_2]$ and $\pN(\cL)$, shorthand for normalization together with the 1- and 2-RDM data.} \\
\parbox[t]{0.22\textwidth}{\raggedright $\Tr[HD_2]$} & \parbox[t]{0.72\textwidth}{\raggedright Contraction of a zero-, one-, and two-body operator $H$ with those RDM data, not a full many-electron density matrix trace.} \\
\parbox[t]{0.22\textwidth}{\raggedright $D_2,Q_2,G_2$} & \parbox[t]{0.72\textwidth}{\raggedright Two-particle, two-hole, and particle--hole positivity blocks.} \\
\parbox[t]{0.22\textwidth}{\raggedright $D_3,E_3,F_3,Q_3$} & \parbox[t]{0.72\textwidth}{\raggedright Third-order blocks with annihilation/creation counts $(3,0)$, $(2,1)$, $(1,2)$, and $(0,3)$ in a row.} \\
\parbox[t]{0.22\textwidth}{\raggedright $T_1,T_2,T_2'$} & \parbox[t]{0.72\textwidth}{\raggedright Selected two-body positivity conditions obtained by cancellation of higher-body terms; not the complete level-three conditions.} \\
\parbox[t]{0.22\textwidth}{\raggedright $X\succeq0$} & \parbox[t]{0.72\textwidth}{\raggedright The Hermitian matrix $X$ is positive semidefinite.} \\
\end{tabular}
\end{ruledtabular}
\end{table*}

\begin{table*}[t]
\caption{Counting model, dispersion test, and interpolated-RDM construction.}
\label{tab:notation-model}
\begin{ruledtabular}
\begin{tabular}{ll}
\parbox[t]{0.22\textwidth}{\raggedright Symbol} & \parbox[t]{0.72\textwidth}{\raggedright Meaning} \\
\hline
\parbox[t]{0.22\textwidth}{\raggedright $A,B$, $m$} & \parbox[t]{0.72\textwidth}{\raggedright Two disjoint shells, each with $m$ spin orbitals; the model fixes $N=m$ electrons in $r=2m$ orbitals.} \\
\parbox[t]{0.22\textwidth}{\raggedright $K$, $k$, $\lambda$} & \parbox[t]{0.72\textwidth}{\raggedright Shell-$A$ count $K=\sum_{i\in A}n_i$, its integer eigenvalue $k$, and the real trial spectral value $\lambda$.} \\
\parbox[t]{0.22\textwidth}{\raggedright $H_\lambda$} & \parbox[t]{0.72\textwidth}{\raggedright Dispersion Hamiltonian $(K-\lambda\id)^2$.} \\
\parbox[t]{0.22\textwidth}{\raggedright $E_{\mathrm{exact}}(\lambda)$} & \parbox[t]{0.72\textwidth}{\raggedright Exact minimum dispersion $\min_{0\le k\le m,\,k\in\mathbb Z}(k-\lambda)^2$.} \\
\parbox[t]{0.22\textwidth}{\raggedright $E_q(\lambda)$, $E_{(2,q)}(\lambda)$} & \parbox[t]{0.72\textwidth}{\raggedright Minimum dispersion under full $q$-positivity and full $(2,q)$-positivity; these values are equal.} \\
\parbox[t]{0.22\textwidth}{\raggedright $Z_q(K)$} & \parbox[t]{0.72\textwidth}{\raggedright Accepted spectral values $\{\lambda:E_q(\lambda)=0\}$, including any spurious zeros.} \\
\parbox[t]{0.22\textwidth}{\raggedright $\rho_k$} & \parbox[t]{0.72\textwidth}{\raggedright Uniform mixture of determinants with $K=k$ and $N=m$.} \\
\parbox[t]{0.22\textwidth}{\raggedright $\widehat{\mathcal L}_\lambda$} & \parbox[t]{0.72\textwidth}{\raggedright Fixed-$N$ signed interpolation, \cref{eq:fixedN-interpolation}.} \\
\parbox[t]{0.22\textwidth}{\raggedright $M$, $\theta$} & \parbox[t]{0.72\textwidth}{\raggedright Orbital count and real continued particle number of an auxiliary single shell.} \\
\parbox[t]{0.22\textwidth}{\raggedright $\fall{z}{s}$} & \parbox[t]{0.72\textwidth}{\raggedright Falling factorial $z(z-1)\cdots(z-s+1)$, with $\fall{z}{0}=1$.} \\
\parbox[t]{0.22\textwidth}{\raggedright $n_S$} & \parbox[t]{0.72\textwidth}{\raggedright Occupation product $\prod_{i\in S}n_i$; $n_\varnothing=\id$.} \\
\parbox[t]{0.22\textwidth}{\raggedright $\langle\cdot\rangle_{\theta,M}$} & \parbox[t]{0.72\textwidth}{\raggedright Single-shell interpolated expectation with $\langle n_S\rangle_{\theta,M}=\fall{\theta}{|S|}/\fall{M}{|S|}$.} \\
\parbox[t]{0.22\textwidth}{\raggedright $\langle\cdot\rangle_\lambda$} & \parbox[t]{0.72\textwidth}{\raggedright Fixed-$N$ interpolated expectation $\widehat{\mathcal L}_\lambda$, evaluated by the product formula through occupation degree $m$.} \\
\parbox[t]{0.22\textwidth}{\raggedright $\langle\cdot\rangle_{\lambda,\mathrm{prod}}$} & \parbox[t]{0.72\textwidth}{\raggedright Auxiliary product of two shell functionals, \cref{eq:pseudo}; not a fixed-$N$ functional at all degrees.} \\
\parbox[t]{0.22\textwidth}{\raggedright $U,V$, $\chi_{U,V}$} & \parbox[t]{0.72\textwidth}{\raggedright Disjoint sets of removed and added electrons in a row; $\chi_{U,V}=\prod_{i\in U}n_i\prod_{j\in V}(1-n_j)$ is its required occupation pattern.} \\
\parbox[t]{0.22\textwidth}{\raggedright $M_A,M_B$} & \parbox[t]{0.72\textwidth}{\raggedright Shell occupation-correlation matrices; $M$ instead denotes a shell size.} \\
\parbox[t]{0.22\textwidth}{\raggedright $\mathbb E_{\mathrm{diag}}$} & \parbox[t]{0.72\textwidth}{\raggedright Orbital-phase average, retaining only the occupation-basis diagonal part of an operator.} \\
\parbox[t]{0.22\textwidth}{\raggedright $\mathcal B_q^{[N]}$} & \parbox[t]{0.72\textwidth}{\raggedright Classical degree-$(q-1)$ positivity plus nonnegative patterns on at most $q$ orbitals.} \\
\parbox[t]{0.22\textwidth}{\raggedright $B_j$} & \parbox[t]{0.72\textwidth}{\raggedright Integer-spectrum inequality operator $(K-j)(K-j-1)$, nonnegative on physical states.} \\
\end{tabular}
\end{ruledtabular}
\end{table*}

\begin{table*}[t]
\caption{Projection, positivity cones, and semidefinite duality.  A cone
allows unnormalized moments; a feasible set $\cF$ includes normalization.}
\label{tab:notation-cones}
\begin{ruledtabular}
\begin{tabular}{ll}
\parbox[t]{0.22\textwidth}{\raggedright Symbol} & \parbox[t]{0.72\textwidth}{\raggedright Meaning} \\
\hline
\parbox[t]{0.22\textwidth}{\raggedright $\cV_2$} & \parbox[t]{0.72\textwidth}{\raggedright Space of zero-, one-, and two-body operators on the fixed-$N$ space.} \\
\parbox[t]{0.22\textwidth}{\raggedright $\cV_q$, $\cV_q^*$} & \parbox[t]{0.72\textwidth}{\raggedright Real space of Hermitian operators of body order at most $q$ on $\cH_N$, modulo operator identities, and its dual space of expectation assignments.} \\
\parbox[t]{0.22\textwidth}{\raggedright $\pN$, $\pNadj$} & \parbox[t]{0.72\textwidth}{\raggedright Contraction to two-body data and its adjoint lifting into the retained operator space, with fixed-$N$ identities understood.} \\
\parbox[t]{0.22\textwidth}{\raggedright $\cF_q^{\mathrm{full},[N]}$} & \parbox[t]{0.72\textwidth}{\raggedright Normalized moment assignments satisfying all particle--hole positivity blocks through level $q$.} \\
\parbox[t]{0.22\textwidth}{\raggedright $\cF_{2,q}^{\mathrm{polar},[N]}$} & \parbox[t]{0.72\textwidth}{\raggedright Normalized two-body data satisfying every $(2,q)$ condition; equal to $\pN(\cF_q^{\mathrm{full},[N]})$.} \\
\parbox[t]{0.22\textwidth}{\raggedright $\cF_{2,q}^{\mathrm{listed},[N]}$} & \parbox[t]{0.72\textwidth}{\raggedright Two-body data satisfying only the listed classes of $(2,q)$ conditions, potentially a larger feasible set.} \\
\parbox[t]{0.22\textwidth}{\raggedright $\SigmaN$} & \parbox[t]{0.72\textwidth}{\raggedright Cone of sums of squares $C^\dagger C$ with $C\in\cW_{\le q}$, after quotienting operator identities.} \\
\parbox[t]{0.22\textwidth}{\raggedright $\coneN$} & \parbox[t]{0.72\textwidth}{\raggedright Cone of two-body operators whose embeddings belong to $\SigmaN$: the complete $(2,q)$ inequalities.} \\
\parbox[t]{0.22\textwidth}{\raggedright $\cQ_q^{[N]}$, $S$} & \parbox[t]{0.72\textwidth}{\raggedright Unnormalized full $q$-positive moment cone and its two-body projection $S=\pN(\cQ_q^{[N]})$.  This $S$ is a cone, not an orbital subset in $n_S$.} \\
\parbox[t]{0.22\textwidth}{\raggedright $\cP_{2,q}^{\mathrm{polar},[N]}$} & \parbox[t]{0.72\textwidth}{\raggedright Unnormalized two-body moment cone dual to $\coneN$.} \\
\parbox[t]{0.22\textwidth}{\raggedright $\mathcal K^*$, $\overline{\mathcal K}$} & \parbox[t]{0.72\textwidth}{\raggedright For a cone $\mathcal K$, its dual under nonnegative pairing and its closure.  On a linear map, $*$ instead denotes the adjoint.} \\
\parbox[t]{0.22\textwidth}{\raggedright $y$, $\cM_{q,\alpha}(y)$} & \parbox[t]{0.72\textwidth}{\raggedright Moment vector and its particle--hole block; $\alpha$ labels the block.  The index $q$ is sometimes suppressed.} \\
\parbox[t]{0.22\textwidth}{\raggedright $L$, $z$} & \parbox[t]{0.72\textwidth}{\raggedright Map of homogeneous operator identities $Ly=0$ and their free dual multipliers; $L^*z$ vanishes on $\cH_N$.  The map $L$ is distinct from the expectation functional $\cL$.} \\
\parbox[t]{0.22\textwidth}{\raggedright $X_\alpha$} & \parbox[t]{0.72\textwidth}{\raggedright Positive semidefinite dual matrix multiplying positivity block $\alpha$.} \\
\parbox[t]{0.22\textwidth}{\raggedright $c$, $E$} & \parbox[t]{0.72\textwidth}{\raggedright Hamiltonian coefficient vector and scalar energy bound in the raw-coordinate dual certificate \cref{eq:raw-dual}.} \\
\end{tabular}
\end{ruledtabular}
\end{table*}

\clearpage
\section{Raw-coordinate form of the polar equality}
\label{app:raw}

If moments are represented in redundant tensor coordinates rather than in
the CAR quotient, we separate the homogeneous CAR and fixed-$N$ identities,
written $Ly=0$ for unnormalized moment vectors $y$ that include the
identity coordinate, from the normalization $y(\id)=1$.  The full
$q$-positive primal for an objective $c$ is
\[
 \begin{aligned}
 &\min_y\ \langle c,y\rangle\\
 &\text{subject to}\quad
 \cM_\alpha(y)\succeq0,\quad Ly=0,\quad y(\id)=1.
 \end{aligned}
\]
Its dual is to maximize $E$ subject to
\begin{equation}
 c-E\,\id=
 \sum_\alpha\cM_\alpha^*(X_\alpha)+L^*z,
 \qquad X_\alpha\succeq0,
 \label{eq:raw-dual}
\end{equation}
where $z$ is free.  Each $L^*z$ is an operator vanishing on $\cH_N$, so
\cref{eq:raw-dual} states that $H-E$ equals a sum of $q$-level squares up to
an operator that vanishes on $\cH_N$.  If instead the normalization is
merged into affine constraints $Ly=b$, the same certificate reads
$c=\sum_\alpha\cM_\alpha^*(X_\alpha)+L^*z$ with dual value $z^{\mathsf T}b$,
and the vanishing operator is $L^*z-(z^{\mathsf T}b)\id$.  For a two-body
objective, the components of \cref{eq:raw-dual} above body order two must
cancel.  In an operator quotient the term $L^*z$ disappears because the
identities have already been factored out.  This explains precisely where contraction, normalization,
antisymmetry, and fixed-number multipliers enter the equivalence in
\cref{prop:projection}.

\section{Proof of the projection--polar equivalence}
\label{app:polar}

This appendix transcribes the argument of Erdahl \cite[Sec.~IV]{Erdahl2007}
to the fixed-$N$ quotient.  The completely mixed state below plays the role
of Erdahl's positive definite $k$-matrix of the normalized identity
\cite[property R1]{Erdahl2007}, which guarantees his existence condition
\cite[Theorem~13]{Erdahl2007}.

Let $\cV_q$ be the real space of Hermitian operators of body order at most
$q$ on $\cH_N$, including the identity and quotienting CAR and fixed-$N$
operator identities.  Its dual $\cV_q^*$ consists of real linear assignments
of expectations.  Thus a moment vector $y$ lies in $\cV_q^*$, whereas a
two-body observable lies in $\cV_2$.  Write the particle--hole Gram map as
$\cM:\cV_q^*\to\bigoplus_\alpha\mathbb H^{d_\alpha}$ and set
\[
 \cQ_q^{[N]}=\{y:\cM(y)\succeq0\}.
\]
The following finite-dimensional facts provide the needed closedness;
neither is inferred merely from being a semidefinite projection.

\paragraph*{Boundedness of normalized moments.}
In an antisymmetric index basis, every retained particle block obeys
\[
 D_p\succeq0,\qquad \Tr D_p=\binom Np\,y(\id).
\]
For $y(\id)=1$, this bounds every eigenvalue and hence every entry of $D_p$.
All remaining particle--hole entries are CAR-linear combinations of these
normal-ordered moments.  The normalized feasible set is therefore bounded;
it is also closed, since its constraints are linear equalities and matrix
positivity conditions.  Consequently
\[
 T=\pN(\cF_q^{\mathrm{full},[N]})
\]
is compact and convex.  For an unnormalized feasible $y$, positivity gives
$y(\id)\ge0$.  If $y(\id)=0$, every $D_p$ has zero trace and is positive
semidefinite, so all its entries vanish; the CAR identities then give
$y=0$.  Every nonzero member of $\cQ_q^{[N]}$ can thus be normalized by
dividing by $y(\id)>0$.

\paragraph*{Closedness of the sum-of-squares cone.}
In each particle-number-change sector choose independent row operators
$C_{\alpha i}P_N$, removing those that vanish on $\cH_N$.  The completely
mixed state $\rho_*=\id_{\cH_N}/\dim\cH_N$ has positive definite Gram blocks
\[
 (G_\alpha)_{ij}=\Tr(\rho_*C_{\alpha i}^\dagger C_{\alpha j}),
\]
because a nonzero linear combination of these independent maps has
strictly positive squared Hilbert--Schmidt norm.  Let $g>0$ be the smallest
eigenvalue among all the $G_\alpha$.  If a sequence of sums of squares
$A_n=\cM^*(X_n)$ converges, with $X_n=\bigoplus_\alpha X_{n,\alpha}\succeq0$,
then
\[
 \Tr(\rho_* A_n)=\sum_\alpha\Tr(G_\alpha X_{n,\alpha})
 \ge g\sum_\alpha\Tr X_{n,\alpha}.
\]
The left-hand side is bounded, so the positive semidefinite coefficient
matrices $X_n$ are bounded.  A subsequence converges to $X\succeq0$, and
$A_n\to\cM^*(X)$ along it.  The limit is still a sum of squares, proving
that $\SigmaN$ is closed.

\begin{proof}[Proof of \cref{prop:projection} by separation]
The forward inclusion follows by taking expectations of sums of squares.
For the reverse inclusion, let normalized $D_2$ satisfy every $(2,q)$
condition and suppose $D_2\notin T$.  Since $T$ is compact and convex,
finite-dimensional separation gives a two-body operator $H$ and a real
number $c$ such that
\[
 D_2(H)<c\le D'_2(H)\qquad(D'_2\in T).
\]
Put $B_2=H-c\id$.  Normalization then gives $D_2(B_2)<0$, whereas
$y(\pNadj B_2)\ge0$ for every normalized full $q$-positive assignment.
By the normalization argument above, this inequality holds for every
$y\in\cQ_q^{[N]}$.

It remains to show that $\pNadj B_2$ is a sum of squares.  If it were
outside the closed convex cone $\SigmaN$, separation from that cone would
give a linear expectation assignment $y$ with
\[
 y(A)\ge0\quad(A\in\SigmaN),\qquad y(\pNadj B_2)<0.
\]
The first condition is exactly positivity on all $q$-level squares, namely
$y\in\cQ_q^{[N]}$, contradicting the preceding paragraph.  Hence
$\pNadj B_2\in\SigmaN$, or $B_2\in\coneN$.  But the assumed $(2,q)$
conditions require $D_2(B_2)\ge0$, the final contradiction.  This proves
\cref{eq:projection}.  Minimizing the same two-body objective over the
equal sets gives \cref{eq:energy-equality}.
\end{proof}

For comparison with cone duality, put $S=\pN(\cQ_q^{[N]})$.  The two
separation steps express the familiar chain
\begin{equation}
 \cP_{2,q}^{\mathrm{polar},[N]}
 \overset{\text{(a)}}{=}\coneN^*
 \overset{\text{(b)}}{=}(S^*)^*
 \overset{\text{(c)}}{=}\overline S
 \overset{\text{(d)}}{=}S .
 \label{eq:bipolar-chain}
\end{equation}
Here (a) is the definition, (b) follows from closedness of $\SigmaN$,
and (c) is the finite-dimensional bipolar theorem.  For (d), the
normalization argument gives $S=\{tD:t\ge0,\ D\in T\}$.  Since $T$ is
compact and $D(\id)=1$ on $T$, any convergent sequence $t_nD_n$ has bounded
$t_n=(t_nD_n)(\id)$ and a convergent subsequence of $D_n$; its limit is
again in $S$.  This verifies the two closedness properties separately.

An example of a duality gap that arises when strict feasibility fails is
given in the Supplementary Material.

\section{Positivity of the interpolated state at arbitrary order}
\label{app:general}

Section~\ref{sec:q3-check} checked the case $q=3$ by direct inspection
of the blocks.  This appendix proves the general case in four steps: the
interpolated state of a single shell and its occupation-pattern probabilities;
the known positivity of its occupation-correlation matrices; the phase
average and the resulting sub-block structure of all particle--hole blocks;
and positivity at arbitrary order.

\subsection{A single shell}

Equation \eqref{eq:pseudo-values} factorizes over the two shells, and each
factor is the same object for a single shell.  Consider a shell of $M$
orbitals containing exactly $k$ particles, in the state that mixes all
$\binom Mk$ determinants with equal weight.  The diagonal elements of its
reduced density matrices are occupation correlations,
\[
 \langle n_S\rangle=\frac{\fall{k}{|S|}}{\fall{M}{|S|}},
 \qquad n_S=\prod_{i\in S}n_i ,
\]
the probability that all orbitals in $S$ are occupied, where
$\fall{z}{s}=z(z-1)\cdots(z-s+1)$ and $\fall{z}{0}=1$.  Replacing $k$ by a
real number $\theta$ defines the shell interpolated expectation
\begin{equation}
 \langle n_S\rangle_{\theta,M}
 =
 \frac{\fall{\theta}{|S|}}{\fall{M}{|S|}}
 =
 \frac{\binom{\theta}{|S|}}{\binom{M}{|S|}},
 \label{eq:grig-functional}
\end{equation}
extended linearly to all polynomials in the occupation numbers, with
$n_i^2=n_i$.  For a shell of two orbitals,
\begin{gather*}
 \langle 1\rangle_{\theta,2}=1,\qquad
 \langle n_1\rangle_{\theta,2}=\langle n_2\rangle_{\theta,2}=\frac{\theta}{2},\\
 \langle n_1n_2\rangle_{\theta,2}=\frac{\theta(\theta-1)}{2}.
\end{gather*}
For $\theta=1$ this is one particle in two orbitals: each orbital is
occupied with probability $1/2$, and never both.  For $\theta=1/2$ the pair
occupation equals $-1/8$, a negative diagonal element of the 2-RDM that no
state can have.

\begin{lemma}[Occupation-pattern probabilities]
\label{lem:conditioning}
Let $U$, $V$, and $T$ be disjoint sets of orbitals in a shell of $M$
orbitals, $u=|U|$, $v=|V|$.  Then
\begin{align}
 &\Bigl\langle \prod_{i\in U}n_i
 \prod_{j\in V}(1-n_j)\,n_T\Bigr\rangle_{\theta,M}\nonumber\\
 &\quad=
 \frac{\fall{\theta}{u}\fall{M-\theta}{v}}
      {\fall{M}{u+v}}\,
 \langle n_T\rangle_{\theta-u,M-u-v}.
 \label{eq:conditioning}
\end{align}
\end{lemma}

For integer $\theta$ this is the familiar statement about conditional
probabilities.  The prefactor is the probability that the $u$ orbitals in
$U$ are occupied and the $v$ orbitals in $V$ are empty.  Given this, the
remaining $\theta-u$ electrons are spread uniformly over the remaining
$M-u-v$ orbitals, which is the ensemble $\langle\cdot\rangle_{\theta-u,M-u-v}$
of a smaller shell.  For example, take two electrons in four spin
orbitals ($\theta=2$, $M=4$; six determinants), $U=\{1\}$, $V=\{2\}$, and
$T=\{3\}$.  Only $|13\rangle$ has orbitals $1$ and $3$ occupied and
orbital $2$ empty, so the left-hand side is $1/6$.  The prefactor is
$\fall{2}{1}\fall{2}{1}/\fall{4}{2}=4/12=1/3$; the remaining electron then
occupies orbital $3$ or $4$ with equal probability, giving
$\langle n_3\rangle_{1,2}=1/2$, and indeed $\frac13\cdot\frac12=\frac16$.
The content of the lemma is that the same rule continues to hold for a
noninteger continued particle number, for instance $\theta=5/2$.

\begin{proof}
Put $s=|T|$.

\emph{Integer $\theta$.}  For $\theta\in\{0,\ldots,M\}$ the left-hand side
is the probability, in the equal mixture of all $\binom M\theta$
determinants with $\theta$ electrons in the shell, that the $u+s$ orbitals
in $U\cup T$ are occupied and the $v$ orbitals in $V$ are empty.  The
qualifying determinants place the remaining $\theta-u-s$ electrons in the
$M-u-v-s$ free orbitals, so
\[
 \frac{\binom{M-u-v-s}{\theta-u-s}}{\binom{M}{\theta}}
 =
 \frac{\fall{\theta}{u+s}\fall{M-\theta}{v}}
      {\fall{M}{u+v+s}}.
\]

\emph{Noninteger $\theta$.}  There is no ensemble to count, but both sides
are polynomials in $\theta$.  On the left, expanding each hole factor
$1-n_j$ writes the expectation as a signed sum of values
$\langle n_W\rangle_{\theta,M}=\fall{\theta}{|W|}/\fall{M}{|W|}$ with
$|W|\le u+v+s\le M$, each a polynomial in $\theta$ of degree at most $M$;
the closed form above is also such a polynomial.  A nonzero polynomial of
degree at most $M$ has at most $M$ roots, so two such polynomials that agree
at the $M+1$ points $\theta=0,\ldots,M$ are identical.  Hence the closed
form holds for every real $\theta$.

\emph{Right-hand side.}  The right-hand side is
\[
 \frac{\fall{\theta}{u}\fall{M-\theta}{v}}{\fall{M}{u+v}}
 \cdot\frac{\fall{\theta-u}{s}}{\fall{M-u-v}{s}},
\]
and since a falling factorial continued from where it stopped is a longer
falling factorial, $\fall{\theta}{u}\fall{\theta-u}{s}=\fall{\theta}{u+s}$
and $\fall{M}{u+v}\fall{M-u-v}{s}=\fall{M}{u+v+s}$, this equals the same
closed form.
\end{proof}

For the two orbital groups define an auxiliary product functional on
occupation numbers in the full Fock space, for every set $S$ of orbitals,
\begin{equation}
 \langle n_S\rangle_{\lambda,\mathrm{prod}}
 =
 \frac{\fall{\lambda}{|S\cap A|}}
      {\fall{m}{|S\cap A|}}
 \frac{\fall{m-\lambda}{|S\cap B|}}
      {\fall{m}{|S\cap B|}}.
 \label{eq:pseudo}
\end{equation}
Thus
\[
 \langle\cdot\rangle_{\lambda,\mathrm{prod}}
 =
 \langle\cdot\rangle_{\lambda,m}\otimes\langle\cdot\rangle_{m-\lambda,m}
\]
on the commuting occupation numbers, and \cref{lem:conditioning} with
$T=\varnothing$, applied in each shell, gives the occupation-pattern
probabilities \cref{eq:pseudo-values}.  The two continued particle numbers add
to $m$, so the fixed-$N$ identities
$\langle(\widehat N-N)n_S\rangle_{\lambda,\mathrm{prod}}=0$ hold whenever $S$ contains fewer
than $m$ orbitals of each shell.  They do not hold at all degrees: for
$n_A=\prod_{i\in A}n_i$, the value is
$(m-\lambda)\fall{\lambda}{m}/m!$, generally nonzero.  The fixed-$N$
functional $\langle\cdot\rangle_\lambda=\widehat{\mathcal L}_\lambda$ is instead
\cref{eq:fixedN-interpolation}; it agrees with this auxiliary formula
through total occupation degree $m$, covering all moments needed for
$\mathcal B_q^{[N]}$ on the central interval.

\subsection{Occupation-correlation matrices}

The level-$q$ constraints do not test individual entries, however, but
matrices of occupation correlations.  The simplest is the number-operator part
of $G_2$ in \cref{eq:pseudo-blocks}: on the operators $1,n_1,\ldots,n_M$ it
is the $(M+1)\times(M+1)$ matrix of $\langle n_in_k\rangle_{\theta,M}$.
More generally, rows labelled by products $n_S$ of up to $d$ occupation
numbers give the occupation-correlation matrix of order $d$, with entries
$\langle n_Sn_T\rangle_{\theta,M}=\langle n_{S\cup T}\rangle_{\theta,M}$, which
are diagonal elements of reduced density matrices of order up to $2d$.  The
question is when these matrices are positive semidefinite even though some
entries, like the pair occupation above, are negative.

\begin{lemma}[Slice positivity]
\label{lem:slice-positive}
Consider a shell of $M$ orbitals with continued particle number $\theta$, and
let $d$ be a nonnegative integer with $2d\le M$.  If
\begin{equation}
 d\le\theta\le M-d,
 \label{eq:slice-central}
\end{equation}
then
\[
 \langle p^\dagger p\rangle_{\theta,M}\ge0
\]
for every operator $p=\sum_{|S|\le d}c_Sn_S$ with complex coefficients; that
is, the occupation-correlation matrix of order $d$ is positive
semidefinite.  Moreover, the shell particle number
$\widehat N_M=\sum_{i=1}^Mn_i$ has zero variance in the sense that
\[
 \bigl\langle(\widehat N_M-\theta)\,p\bigr\rangle_{\theta,M}=0
\]
for every polynomial $p$ in the occupation numbers of degree at most $M-1$,
in particular of degree at most $2d-1$.
\end{lemma}

For example, with $n_i=a_i^\dagger a_i$, the case $d=1$ concerns
\begin{gather*}
 p=c_0+\sum_ic_i\,a_i^\dagger a_i ,\\
 \langle p^\dagger p\rangle_{\theta,M}
 =\sum_{i,k=0}^{M}\overline{c_i}c_k\,
 \langle a_i^\dagger a_i\,a_k^\dagger a_k\rangle_{\theta,M},
\end{gather*}
with the convention $a_0^\dagger a_0=1$; the matrix of coefficients is the
number-operator part of $G_2$.  The case $d=2$ adds the pair operators
$n_in_j=a_i^\dagger a_j^\dagger a_ja_i$,
\[
 p=c_0+\sum_ic_i\,a_i^\dagger a_i
 +\sum_{i<j}c_{ij}\,a_i^\dagger a_j^\dagger a_ja_i ,
\]
and $\langle p^\dagger p\rangle_{\theta,M}$ then involves diagonal elements
of the 4-RDM, which first enter at $q=4$.  The case $d=3$ adds the triple
operators $n_in_jn_k=a_i^\dagger a_j^\dagger a_k^\dagger a_ka_ja_i$,
\[
 p=c_0+\sum_ic_i\,n_i+\sum_{i<j}c_{ij}\,n_in_j
 +\sum_{i<j<k}c_{ijk}\,a_i^\dagger a_j^\dagger a_k^\dagger a_ka_ja_i ,
\]
so that $\langle p^\dagger p\rangle_{\theta,M}$ involves diagonal elements of
the 6-RDM; these first enter at $q=6$, and \cref{lem:slice-positive} then
requires $M\ge6$ and $3\le\theta\le M-3$.  The zero-variance statement with
$p=1$ gives $\langle\widehat N_M\rangle_{\theta,M}=\theta$; with $p=n_i$ it
gives
\[
 \sum_{j\ne i}\langle n_in_j\rangle_{\theta,M}
 =(\theta-1)\langle n_i\rangle_{\theta,M},
\]
the usual contraction of the 2-RDM to the 1-RDM, with $N-1$ replaced by
$\theta-1$.

In words: if the continued particle number lies at least $d$ away from both
the empty and the filled shell, the occupation-correlation matrices through
order $d$ are positive semidefinite, so the corresponding constraints cannot
distinguish the interpolated state from a genuine state, and the shell particle
number is $\theta$ with zero variance.  For integer $\theta$ this is just
positivity of a genuine state; the content is the noninteger case.  The
interpolated state \cref{eq:grig-functional} is the standard one used for
sum-of-squares lower bounds on $0/1$ problems with a fixed number of ones
(the knapsack problem) \cite{Grigoriev2001,Lee2016,Potechin2019}, and
neither it nor \cref{lem:slice-positive} is new; for completeness,
Appendix~\ref{app:slice} gives a self-contained proof from the
Young--Fourier basis of Filmus \cite[Theorems 9, 10, and 15]{Filmus2016}.
Below we lift it to all particle--hole blocks of the fermionic problem;
level $q$ will require the order $d=q-1$.

\begin{remark}[Two orbitals]
\label{rem:slice-two}
Take $M=2$ and $d=1$.  The shell has two orbitals with occupation-number
operators $n_1=a_1^\dagger a_1$ and $n_2=a_2^\dagger a_2$, and the test
operators are
\[
 p=c_0+c_1n_1+c_2n_2=c_0+c_1a_1^\dagger a_1+c_2a_2^\dagger a_2 .
\]
Expanding $p^\dagger p$ with $n_i^2=n_i$ gives
$\langle p^\dagger p\rangle_{\theta,2}=v^\dagger Q v$ with
$v=(c_0,c_1,c_2)^{\mathsf T}$, where
\[
 Q=
 \begin{pmatrix}
 1 & \theta/2 & \theta/2\\
 \theta/2 & \theta/2 & \theta(\theta-1)/2\\
 \theta/2 & \theta(\theta-1)/2 & \theta/2
 \end{pmatrix}
\]
is the Gram matrix of the operators $1,n_1,n_2$, that is, the number-operator
part of $G_2$ for this shell.  Its eigenvalues are $0$,
$\theta(2-\theta)/2$, and $1+\theta^2/2$.  Thus $Q\succeq0$ for all
$0\le\theta\le2$, so $G_2$ does not detect the negative pair occupation.
The null eigenvector is $(-\theta,1,1)^{\mathsf T}$, that is,
$p=n_1+n_2-\theta=a_1^\dagger a_1+a_2^\dagger a_2-\theta$: the zero eigenvalue is precisely the zero-variance
statement of \cref{lem:slice-positive}.
The positive range here, $0\le\theta\le2$, is wider than the range
$\theta=1$ guaranteed by \cref{eq:slice-central}.  In general, for $d\ge1$, the norm factors
of Appendix~\ref{app:slice} show that a noninteger $\theta$ is admissible
exactly when $d-1<\theta<M-d+1$ (for $d=0$ only constants are squared, and
no restriction on $\theta$ arises); the weaker condition
\eqref{eq:slice-central} is all that the proof of \cref{thm:main} requires,
because the exclusion of the remaining values is established separately
in Sec.~\ref{sec:necessity}.
\end{remark}

\subsection{Phase average and sub-block structure}

The product evaluation formula \eqref{eq:pseudo} assigns values to
occupation products; through degree $m$ these equal the fixed-$N$
interpolated expectations.  The positivity blocks, however, also contain entries such as
$\langle a_1^\dagger a_2\rangle$ in $G_2$, $\langle a_1^\dagger
a_2^\dagger a_4a_3\rangle$ in $D_2$, or $\langle a_1^\dagger a_2a_2^\dagger
a_1\rangle$, which is diagonal but not written as a product of occupation
numbers.  Since every $\rho_k$ in \cref{eq:fixedN-interpolation} is diagonal,
only the occupation-basis diagonal part contributes.  For retained block
entries its degree is at most $q\le m$, so it can be evaluated with
\cref{eq:pseudo}.  Rotating the phase of each orbital multiplies a creation
operator by a phase, $e^{i\phi_jn_j}a_j^\dagger e^{-i\phi_jn_j}
=e^{i\phi_j}a_j^\dagger$, so averaging over all phases removes every
operator that changes the occupation of some orbital and keeps the rest.
Formally, let $\mathbb E_{\mathrm{diag}}$ be this phase average,
\begin{equation}
 \mathbb E_{\mathrm{diag}}(O)
 =
 \int_{[0,2\pi]^{2m}}
 e^{i\sum_j\phi_jn_j}Oe^{-i\sum_j\phi_jn_j}
 \frac{d\phi}{(2\pi)^{2m}},
 \label{eq:twirl}
\end{equation}
and the fixed-$N$ interpolation obeys
\[
 \langle O\rangle_\lambda
 =
 \langle \mathbb E_{\mathrm{diag}}O\rangle_\lambda.
\]
For example,
\begin{align*}
 \mathbb E_{\mathrm{diag}}(a_1^\dagger a_2)&=0,\\
 \mathbb E_{\mathrm{diag}}(a_1^\dagger a_2^\dagger a_4a_3)&=0,\\
 \mathbb E_{\mathrm{diag}}(a_1^\dagger a_1)&=n_1,\\
 \mathbb E_{\mathrm{diag}}(a_1^\dagger a_2a_2^\dagger a_1)&=n_1(1-n_2),
\end{align*}
so $\langle a_1^\dagger a_2\rangle_\lambda
=\langle a_1^\dagger a_2^\dagger a_4a_3\rangle_\lambda=0$,
$\langle a_1^\dagger a_1\rangle_\lambda=\langle n_1\rangle_\lambda$, and
$\langle a_1^\dagger a_2a_2^\dagger a_1\rangle_\lambda
=\langle n_1(1-n_2)\rangle_\lambda$.  The last example shows that a
diagonal entry becomes a joint electron--hole probability, whose value is
given by \cref{eq:pseudo-values}.  Thus \cref{eq:twirl} sets all
off-diagonal occupation-basis matrix elements to zero, justifies the block
formulas \cref{eq:pseudo-blocks,eq:pseudo-blocks3,eq:pseudo-blocks3-repeated},
and leaves $\langle(K-\lambda)^2\rangle_\lambda$ unchanged because $K$ is
diagonal.

What must be shown is that every positivity block of these interpolated reduced
density matrices is positive semidefinite.  Each block is the matrix
$\langle C^\dagger C'\rangle$ for a set of row operators $C,C'$, products of
creation and annihilation operators; for instance, $D_2$, $Q_2$, and $G_2$
use the rows $a_ia_j$, $a_i^\dagger a_j^\dagger$, and $a_i^\dagger a_j$.  For a
genuine state these blocks are already separated by the change in total
particle number.  For the diagonal interpolated state they split further, because
two rows give a nonzero entry only if they change the occupation of every
orbital in the same way.  The following lemma
records the resulting sub-blocks.

\begin{lemma}[Sub-block structure]
\label{lem:car-sector}
Let $\langle\cdot\rangle$ be diagonal, $\langle O\rangle
=\langle\mathbb E_{\mathrm{diag}}O\rangle$.  For disjoint sets $U$ and $V$
of orbitals, the row operators of length at most $q$ that remove one
electron from each orbital of $U$ and add one electron to each orbital of $V$
span the same space as
\begin{equation}
 w_{U,V,P}=a_V^\dagger a_U\,n_P,
 \qquad P\cap(U\cup V)=\varnothing,
 \label{eq:sector-normal}
\end{equation}
where $a_U$ and $a_V^\dagger$ are ordered products over $U$ and $V$,
$n_P=\prod_{i\in P}n_i$, and
\begin{equation}
 |U|+|V|+2|P|\le q.
 \label{eq:degree-budget}
\end{equation}
Entries between rows with different $(U,V)$ vanish, and within the
$(U,V)$ sub-block
\begin{equation}
 \langle w_{U,V,P}^\dagger w_{U,V,P'}\rangle
 =
 \Bigl\langle \prod_{i\in U}n_i\prod_{j\in V}(1-n_j)\,
 n_{P\cup P'}\Bigr\rangle .
 \label{eq:sector-square}
\end{equation}
Hence $\langle\cdot\rangle$ is full positive through order $q$ if and only
if every $(U,V)$ sub-block, the matrix \cref{eq:sector-square} with rows
and columns labelled by the admissible sets $P$, is positive semidefinite.
\end{lemma}

\begin{proof}
Normal ordering a product pairs every orbital that appears once as a
creation and once as an annihilation operator into an occupation number, or into $1-n_i$ when
the annihilator stands first.  The unpaired annihilation and creation operators define
$U$ and $V$, and the paired orbitals give a polynomial in the remaining
occupation numbers, that is, a combination of the $n_P$; the length count
gives \cref{eq:degree-budget}.  Rows with different $(U,V)$ acquire
different phases under \cref{eq:twirl}, so their entries vanish.  Within a
sub-block, $(a_V^\dagger a_U)^\dagger(a_V^\dagger a_U)
=\prod_{i\in U}n_i\prod_{j\in V}(1-n_j)$, and the occupation numbers in
$P$ and $P'$ commute with it, which gives \cref{eq:sector-square}.  Since
every test operator is a combination of such rows, positivity of all blocks
is equivalent to positivity of all sub-blocks.
\end{proof}

For $q=2$ the lemma recovers the familiar blocks:
$(U,V)=(\{i,j\},\varnothing)$ with the single row $a_ja_i$ is a diagonal
entry of $D_2$, $(\varnothing,\{i,j\})$ one of $Q_2$, and
$(\{j\},\{i\})$ a diagonal entry of the single-excitation part of $G_2$; these
sub-blocks are $1\times1$.  The sub-block $(\varnothing,\varnothing)$ has
the rows $1,n_1,\ldots,n_r$ and is the number-operator part of $G_2$.  At
$q=3$, the sub-block $(\{s\},\varnothing)$ has the rows $a_s$ and $a_sn_p$
and is the repeated-index part of $E_3$ in
\cref{eq:pseudo-blocks3-repeated}.  These are exactly the ``weight'' and
``matrix'' entries of \cref{tab:block-construction}.

Every entry \cref{eq:sector-square} is a joint electron--hole probability:
an electron is required in every orbital of $U$ and of $P\cup P'$, and a
hole in every orbital of $V$, while the other orbitals are not prescribed.
For two electrons in four spin orbitals with $U=\{1\}$ and $V=\{2\}$, the
entry with $P\cup P'=\varnothing$ is the probability $2/6=1/3$ of
$|13\rangle$ or $|14\rangle$, and the entry with $P\cup P'=\{3\}$ is the
probability $1/6$ of $|13\rangle$ alone.  For a combination
$p=\sum_Pc_Pn_P$ of rows, the quadratic form of the sub-block is
$\langle\chi_{U,V}|p|^2\rangle$ with
$\chi_{U,V}=\prod_{i\in U}n_i\prod_{j\in V}(1-n_j)$.

\subsection{Positivity at arbitrary order}

\begin{lemma}[Positivity at arbitrary order]
\label{lem:general-positivity}
Let $2\le q\le m$ and $q-1\le\lambda\le m-q+1$.  Then
$\langle\cdot\rangle_\lambda$ is full positive through order $q$.
\end{lemma}

\begin{proof}
By \cref{lem:car-sector} it suffices to show that every $(U,V)$ sub-block
is positive semidefinite.  Let $s=|U|+|V|$ and let $d$ be the largest $|P|$
allowed by \cref{eq:degree-budget}, so that $s+2d\le q$.

If $d=0$, the sub-block is the single probability
$\langle\chi_{U,V}\rangle_\lambda$ of a pattern on $s\le q$ orbitals.  By
\cref{eq:pseudo-values} it is a product of factors
$\lambda-j$ and $m-\lambda-j$ with $j\le s-1\le q-1$, divided by positive
numbers, and every such factor is nonnegative because
$q-1\le\lambda\le m-q+1$.

If $d\ge1$, let $p=\sum_{|P|\le d}c_Pn_P$.  All operators involved are
diagonal and commute, and $\chi_{U,V}^2=\chi_{U,V}$, so the quadratic form
is
\[
 \langle\chi_{U,V}|p|^2\rangle_\lambda=\langle|\chi_{U,V}p|^2\rangle_\lambda .
\]
The occupation-number polynomial $f=\chi_{U,V}p$ has degree at most
$s+d\le q-d\le q-1$.  Expand $f$ in products $n_{P_A}n_{P_B}$ with
$P_A\subseteq A$, $P_B\subseteq B$, and $|P_A|,|P_B|\le q-1$.
The total degree of $|f|^2$ is at most $2q-2\le m$, so its fixed-$N$
expectation equals the auxiliary product evaluation.  Thus, by
\cref{eq:pseudo}, $\langle|f|^2\rangle_\lambda=c^\dagger(M_A\otimes M_B)c$,
where $c$ is the coefficient vector and $M_A$, $M_B$ are the
occupation-correlation matrices of order $q-1$ of the two shells, with
continued particle numbers $\lambda$ and $m-\lambda$.  The central interval
gives $2(q-1)\le m$ and $q-1\le\lambda,\,m-\lambda\le m-q+1$, so
\cref{lem:slice-positive} with $d=q-1$ shows $M_A\succeq0$ and
$M_B\succeq0$.  Hence $M_A\otimes M_B\succeq0$ and the quadratic form is
nonnegative.
\end{proof}

The proof shows more generally that any diagonal functional with
positive semidefinite occupation-correlation matrices of order $q-1$ and
nonnegative pattern probabilities on at most $q$ orbitals is full positive
through order $q$; this is used in \cref{prop:classical}.

\section{Young--Fourier proof of slice positivity}
\label{app:slice}

We prove \cref{lem:slice-positive} with the exact normalization used in the
manuscript.  Only the commuting occupation numbers $n_1,\ldots,n_M$ of one
shell enter, so the argument is the classical one for $0/1$ variables.  All
products below are multilinearized with the relations $n_i^2=n_i$.  Let $\cP_{\le d}$ be the complex vector space of
multilinear polynomials of degree at most $d$, put
$S=n_1+\cdots+n_M$, and define the Hermitian form
\begin{equation}
 \langle f,g\rangle_\theta
 =\langle \overline f g\rangle_{\theta,M}.
 \label{eq:pseudo-inner-product}
\end{equation}
The factorial-moment definition gives the truncated slice identity
\begin{equation}
 \langle (S-\theta)h\rangle_{\theta,M}=0
 \qquad (\deg h\le 2d-1).
 \label{eq:truncated-slice-identity}
\end{equation}
Consequently,
$(S-\theta)\cP_{\le d-1}$ lies in the radical of
\cref{eq:pseudo-inner-product}.

Let $\{\chi_B:B\in\cB_{M,j}\}$ be the degree-$j$ top-set basis of
Filmus.  The basis is orthogonal for every exchangeable measure, and its
norm is reduced to that of
\begin{equation}
 \chi_j=\prod_{\ell=1}^{j}(n_{2\ell-1}-n_{2\ell})
\end{equation}
by Theorems 9 and 10 of Ref.~\cite{Filmus2016}.  The same orthogonality
calculation uses only exchangeability of the moments, and hence applies to
$\langle\cdot\rangle_{\theta,M}$.  Moreover,
\begin{equation}
 \chi_j^2
 =\prod_{\ell=1}^{j}
 \bigl[n_{2\ell-1}(1-n_{2\ell})
      +(1-n_{2\ell-1})n_{2\ell}\bigr].
\end{equation}
The product is the sum of $2^j$ disjoint occupation patterns, each fixing $j$
variables to one and $j$ variables to zero.  \Cref{lem:conditioning}
therefore gives the exact algebraic norm formula
\begin{equation}
 \langle\chi_B,\chi_B\rangle_\theta
 =c_B\,2^j
 \frac{\fall{\theta}{j}\fall{M-\theta}{j}}
      {\fall{M}{2j}},
 \qquad c_B>0.
 \label{eq:filmus-norm}
\end{equation}
For integer $\theta$ this is also the slice norm in Theorem 15 of
Ref.~\cite{Filmus2016}.

It remains to justify that these orthogonal directions exhaust the
non-null quotient of $\cP_{\le d}$.  If $\theta$ is an integer satisfying
$d\le\theta\le M-d$, positivity follows immediately from the uniform
probability measure on the Hamming slice of weight $\theta$; Theorem 15 of
Ref.~\cite{Filmus2016} also gives the corresponding top-set basis.  Suppose
therefore that $\theta$ is noninteger.  As a function on the Boolean cube,
$S-\theta$ is nowhere zero.  Multiplication by $S-\theta$ is consequently
injective in the Boolean function algebra, and in particular on
$\cP_{\le d-1}$.  Hence
\begin{align}
 \dim\frac{\cP_{\le d}}
 {(S-\theta)\cP_{\le d-1}}
 &=\sum_{t=0}^{d}\binom Mt-
   \sum_{t=0}^{d-1}\binom Mt
 =\binom Md.                                      \label{eq:quotient-dim}
\end{align}
The number of top-set basis elements of degrees $0,\ldots,d$ is also
\begin{equation}
 \sum_{j=0}^{d}
 \left[\binom Mj-\binom M{j-1}\right]
 =\binom Md,
\end{equation}
where $\binom M{-1}=0$.  Under
$d\le\theta\le M-d$, every norm in
\cref{eq:filmus-norm} with $j\le d$ is strictly positive for noninteger
$\theta$.  Their classes are therefore linearly independent and form an
orthogonal basis of the quotient in \cref{eq:quotient-dim}.  For every
$p\in\cP_{\le d}$ we may thus write, modulo the radical,
\begin{equation}
 p\equiv
 \sum_{j=0}^{d}\sum_{B\in\cB_{M,j}}
 \widehat p(B)\chi_B.
 \label{eq:young-expansion}
\end{equation}
Combining orthogonality with \cref{eq:filmus-norm} yields
\begin{align}
 \langle |p|^2\rangle_{\theta,M}
 ={}&\sum_{j=0}^{d}
 2^j\frac{\fall{\theta}{j}\fall{M-\theta}{j}}
          {\fall{M}{2j}}
 \sum_{B\in\cB_{M,j}}c_B|\widehat p(B)|^2
 \ge0.                                             \label{eq:pseudo-norm-decomposition}
\end{align}
This proves the first assertion of \cref{lem:slice-positive}, including
complex coefficients.  Finally, for every square-free monomial $n_T$ with $|T|\le M-1$,
\begin{align}
 \langle (S-\theta)n_T\rangle_{\theta,M}
 ={}&(|T|-\theta)
 \frac{\fall{\theta}{|T|}}{\fall{M}{|T|}}\nonumber\\
 &+(M-|T|)
 \frac{\fall{\theta}{|T|+1}}{\fall{M}{|T|+1}}
 =0.                                               \label{eq:slice-ideal-check}
\end{align}
which proves the second assertion by linearity.  The restriction
$|T|\le M-1$ is necessary: for $|T|=M$ the relation $n_i^2=n_i$ gives
$(S-\theta)n_T=(M-\theta)n_T$, whose value
$(M-\theta)\fall{\theta}{M}/M!$ is nonzero for noninteger $\theta$.

%

\end{document}